%% file: TAES2026.tex
\documentclass{IEEEtaes}

\usepackage{color,amsthm}
\usepackage{cite}
\usepackage{amsmath,amssymb,amsfonts}
\usepackage{algorithmic}
\usepackage{algorithm}
\usepackage{graphicx}
\usepackage{textcomp}
\usepackage[dvipsnames]{xcolor}
\usepackage{paralist}
\usepackage{eurosym}
\usepackage{flushend}
\usepackage{url}
\usepackage[utf8]{inputenc}
\usepackage{soul}
\usepackage{caption}
\usepackage[list=true]{subcaption}
\usepackage{booktabs}
\usepackage{tabu}
\usepackage{wrapfig}
\usepackage{arydshln}
\usepackage[normalem]{ulem}
\usepackage{multirow}
\usepackage{siunitx}
\usepackage{adjustbox}
\usepackage{array}
\usepackage{diagbox}

\jvol{XX}
\jnum{XX}
\jmonth{September}
\paper{1234567}
\pubyear{2026}
\doiinfo{TAES.2020.Doi Number}

\newtheorem{theorem}{Theorem}
\newtheorem{lemma}{Lemma}
\newtheorem{proposition}{Proposition}
\newtheorem{corollary}{Corollary}
\theoremstyle{definition}
\newtheorem{assumption}{Assumption}
\newtheorem{definition}{Definition}
\newtheorem{remark}{Remark}
\theoremstyle{plain}
\providecommand{\E}{\mathbb{E}}
\providecommand{\Prob}{\mathbb{P}}
\providecommand{\ind}{\mathbf{1}}
\providecommand{\inner}[2]{\langle #1,\,#2\rangle}

\providecommand{\SNR}{\mathrm{SNR}}

\providecommand{\gagg}{\gamma_{\mathrm{agg}}}
\providecommand{\gzero}{\gamma_{0}}

\providecommand{\Coll}{\mathcal{C}}          % colluder set
\providecommand{\nt}{n_{\mathrm{t}}}         % Tardos code length (shared block)
\providecommand{\eKD}{\varepsilon_{\mathrm{KD}}}% RMS feature-distillation residual
\providecommand{\wt}{\operatorname{wt}}

\providecommand{\Img}{\operatorname{Im}}
\providecommand{\HD}{\mathrm{HD}}
\providecommand{\Com}{\mathsf{Com}}

\makeatletter\def\p@subsection{\thesection-}\makeatother
\begin{document}

\title{ZK-Trace: Certified Collusion Tracing with Zero-Knowledge Credentials for Federated GNSS Interference Monitoring}

\input{00_authors.tex}

\maketitle

\input{00_abstract.tex}


\input{01_introduction.tex}
\input{02_related_work.tex}
\input{03_preliminaries.tex}
\input{04_method.tex}
\input{05_theory.tex}
\input{06_experimental_setup.tex}
\input{07_results.tex}
\input{08_conclusion.tex}

\bibliographystyle{IEEEtran}
\bibliography{references}

\input{09_appendix.tex}

\end{document}

%% file: 00_authors.tex
\author{Redwanul Karim}
%\affil{Fraunhofer Institute for Integrated Circuits IIS, 90411 Nürnberg} 

\author{Nisha L. Raichur}
%\affil{Fraunhofer Institute for Integrated Circuits IIS, 90411 Nürnberg} 

\author{Lucas Heublein}
%\affil{Fraunhofer Institute for Integrated Circuits IIS, 90411 Nürnberg} 

\author{Tobias Feigl}
%\affil{Fraunhofer Institute for Integrated Circuits IIS, 90411 Nürnberg} 

\author{Christopher Mutschler}
%\affil{Fraunhofer Institute for Integrated Circuits IIS, 90411 Nürnberg} 

\author{Felix Ott}
\affil{Fraunhofer Institute for Integrated Circuits IIS, 90411 Nürnberg\\University of Technology Nürnberg (UTN), 90461 Nürnberg}

\receiveddate{Manuscript received September 08, 2026; revised XXXXX 00, 0000; accepted XXXXX 00, 0000.\\
This work has been carried out within the DARCII project, funding code 50NA2401, sponsored by the German Federal Ministry for Economic Affairs and Climate Action (BMWK), and supported by the German Space Agency at DLR, the Bundesnetzagentur (BNetzA), and the Federal Agency for Cartography and Geodesy (BKG). We used ChatGPT 5.4 solely for writing improvements.}
\accepteddate{XXXXX XX XXXX}
\publisheddate{XXXXX XX XXXX}

\corresp{{\itshape (Corresponding author: F. Ott)}.}

\authoraddress{All authors are with the Fraunhofer Institute for Integrated Circuits IIS, 90411 Nürnberg, Germany 
(e-mail: \{\href{mailto:redwanul.karim@iis.fraunhofer.de}{redwanul.karim}, \href{mailto:nisha.lakshmana.raichur@iis.fraunhofer.de}{nisha.lakshmana.raichur}, \href{mailto:lucas.heublein@iis.fraunhofer.de}{lucas.heublein}, \href{mailto:tobias.feigl@iis.fraunhofer.de}{tobias.feigl}, \href{mailto:christopher.mutschler@iis.fraunhofer.de}{christopher.mutschler}, \href{mailto:felix.ott@iis.fraunhofer.de}{felix.ott}\}@iis.fraunhofer.de).}

\markboth{KARIM ET AL.}{ZK-Trace: Collusion-Secure Traitor Tracing with Zero-Knowledge Exculpation}

%% file: 00_abstract.tex
\begin{abstract}
Federated global navigation satellite system (GNSS) monitoring distributes a proprietary classifier to partly trusted stations, any of which may leak its copy. ZK-Trace combines public identity marks, recipient-specific Tardos fingerprints, and zero-knowledge credential verification. The registry supports offline tracing without the leaker's cooperation. We establish conditional false-accusation bounds for arbitrary recovered bit patterns, a finite completeness bound under a hidden-bias residual channel, and a deterministic tracing-score bound for correlated feature-distillation errors. An interval-arithmetic checker makes the conditional bound executable and allocates a common budget across accusation and tamper decisions. Under innocent-row independence, the certificate-based evaluation uses a false-naming budget of $10^{-3}$ per investigation. It isolates all 160 single-owner copies and traces 712 of 720 two-owner mixtures without naming an innocent. Experiments use a simulated GNSS federation and CIFAR-10. Feature matching preserves the feature mark in $20/20$ runs and cross-architecture transfer in $19/20$, at copy-accuracy costs of 4.8 and 6.1 percentage points on GNSS and CIFAR-10. Function-only distillation erases the feature mark, and distillation also removes weight-space marks. These results support verifiable tracing under explicit statistical and cryptographic assumptions. Credential knowledge and recipient evidence serve distinct roles.
\end{abstract}

%% file: 01_introduction.tex
\section{INTRODUCTION}
\label{label_introduction}

Federated learning (FL)~\cite{Fedavg} allows sensor stations to train a shared model while keeping their measurements local. This is a natural fit for global navigation satellite system (GNSS) interference monitoring~\cite{heublein_feigl_crpa}, but it also distributes a proprietary classifier to partly trusted participants. A station that leaks its issued controlled-reception-pattern antenna model exposes the operator's detection capabilities. An adversary can then inspect which jamming and spoofing patterns evade detection, with consequences for positioning, navigation, and timing (PNT).

Tracing such a leak requires evidence linking the model to an enrolled station. Existing federated watermarks provide only part of that evidence. Zero-knowledge ownership schemes such as FedZKP~\cite{yang_yin_zhu} certify group membership without identifying a particular client. Per-client fingerprints identify recipients, but are assigned by the server and read in plaintext, without binding to a client secret~\cite{shao_yang_fedtracker,yu_hong_duw}. Collusion adds a further difficulty: stations can combine their copies to weaken individual marks. Collusion-secure fingerprinting codes address this attack~\cite{boneh1998collusion,tardos2008optimal}, but their guarantees do not directly cover the extraction errors introduced by federated averaging (FedAvg).

ZK-Trace separates two questions: which credentials are represented in the shared model, and which enrolled recipient is associated with a leaked copy (Fig.~\ref{figure_method_overview}). Public identity codewords address the first question, while a recipient-specific Tardos fingerprint~\cite{tardos2008optimal,skoric2008symmetric} addresses the second question. The registry supports offline comparison, and a zero-knowledge proof authenticates the claimant during a dispute. The central analysis turns tracing into an auditable decision: it bounds an innocent score conditional on the recovered bits, then separately bounds missed colluders under a specified channel. The construction supports a BN-scale carrier~\cite{li_fan_gu} and a feature carrier whose score can be stable under feature matching.

We evaluate ZK-Trace for federated few-shot GNSS interference monitoring, where labeled events are scarce~\cite{gaikwad_heublein}. Prototypical networks~\cite{snell_swersky} provide the classifier. Experiments use ten seeds, ten simulated stations partitioned from real GNSS recordings, and CIFAR-10 as a transfer check. This setting tests whether tracing remains useful when data are imbalanced, participation varies, and attackers modify their issued copies.

\textbf{Contributions.}
\textbf{(C1) Finite and executable tracing guarantees.} We give a conditional false-accusation bound that does not assume independent extraction errors (Theorem~\ref{thm:tardos}), together with a finite hidden-bias completeness bound allowing coalition decisions across positions (Theorem~\ref{thm:finite-completeness}). An interval checker enforces the bound and a shared error budget. Exhaustive evaluation of two-owner mixtures traces $712/720$ mixtures and isolates all $160$ single-owner copies without naming an innocent (Sec.~\ref{sec:eval_tracing}).

\textbf{(C2) Credential verification and score stability.} We extend FedZKP's credential mechanism~\cite{yang_yin_zhu} with per-client identity blocks, recipient fingerprints, and a registry. The artifact-bound verifier authenticates a credential statement about a complete model state. It does not claim authorship. A deterministic weighted-score bound establishes when feature perturbations preserve tracing even with correlated or targeted errors (Theorem~\ref{thm:score-stability}).

\textbf{(C3) A GNSS evaluation with explicit robustness limits.} A common attack suite evaluates five baselines on their native mark metrics. Distillation erases all tested weight-space marks. Our feature carrier survives feature-matching distillation and nineteen of twenty cross-architecture runs, but function-only distillation erases it. Its dispatched-copy accuracy cost is $4.8$ percentage points on GNSS and $6.1$ on CIFAR-10. No innocent is accused in the coalition sweep, whereas a DeepMarks-family comparator~\cite{chen_deepmarks} produces false accusations (Sec.~\ref{sec:eval_robustness}).

%% file: 02_related_work.tex
\section{RELATED WORK}
\label{label_related_work}

\textbf{Watermarking and fingerprinting codes.} Model ownership, recipient identification, and survival under model modification are distinct capabilities. White-box watermarks encode information in model parameters, whereas black-box watermarks use queryable triggers. Collusion-secure fingerprinting identifies a source when recipients combine their copies~\cite{boneh1998collusion}. Tardos codes achieve length $O(c^2\log(N/\varepsilon_1))$ for coalition size $c$ and false-accusation target $\varepsilon_1$~\cite{tardos2008optimal}. Symmetric scoring~\cite{skoric2008symmetric} and later refinements improve the score and length analysis~\cite{oosterwijk2013optimal,skoric2008better,laarhoven2014optimal,skoric2015revisited}. Our analysis complements noisy-Tardos work on additive white Gaussian noise~\cite{kuribayashi2010awgn} with conditional score certification and a finite residual-channel bound for model fingerprints.

\textbf{Federated model ownership.} FedIPR~\cite{li_fan_gu} and WAFFLE~\cite{tekgul_xia_marchal} embed ownership marks in models redistributed through FedAvg~\cite{Fedavg}. Their verification establishes mark presence but may expose the secret or depend on its holder. Proofs of knowledge address credential disclosure. Our verification uses the $\Sigma$-protocol~\cite{damgard} for exact-weight Learning Parity with Noise (xLPN) given by Jain et al.~\cite{jain_krenn_pietrzak}, following V\'eron's identification formulation~\cite{veron}. This cryptographic layer authenticates a credential. Recipient tracing additionally requires a per-copy identifier.

\textbf{Per-client attribution.} The closest federated methods provide different parts of the required evidence. FedZKP~\cite{yang_yin_zhu} derives a group watermark from all clients' xLPN public inputs and proves ownership in zero knowledge, but does not identify a recipient. FedTracker~\cite{shao_yang_fedtracker} and DUW~\cite{yu_hong_duw} provide per-client fingerprints assigned by a trusted party, without a collusion-secure code. DeepMarks~\cite{chen_deepmarks} embeds anti-collusion codes centrally, but assumes exact extraction and lacks a false-accusation bound at arbitrary coalition sizes.

Cryptographic participation evidence serves a different purpose. FedPoP~\cite{isler_fedpop} proves participation anonymously and leaves no in-model tracing artifact. FedAaT~\cite{liu_fedaat} adds per-client output-space sequences, but reads them in plaintext and confines zero knowledge to the credential layer. BlackCATT~\cite{rodriguezlois2026blackcatt} uses per-client Tardos labels in a federated trigger channel. It does not establish their composition through the aggregation channel considered here.

\textbf{Distillation and the remaining gap.} Removal attacks include extraction-based erasure~\cite{shafieinejad2021robustness} and knowledge distillation~\cite{hinton_vinyals_dean}, within the taxonomy of~\cite{lukas2022sok}. DAWN~\cite{szyller2021dawn} marks prediction-API responses and evaluates coalition resistance empirically. Entangled watermarks~\cite{jia2021entangled} can survive distillation but do not identify individual leakers. ZK-Trace combines recipient tracing with credential authentication and certificate-gated adjudication (Table~\ref{tab:baselines}). ZK-Trace builds on FedZKP's credential-based watermarking framework~\cite{yang_yin_zhu}, using BN scaling parameters as the embedding carrier, as in FedIPR~\cite{li_fan_gu}. It analyzes Tardos tracing through FedAvg and adds a feature carrier whose survival depends on what the distiller reproduces.

\textbf{Federated few-shot GNSS monitoring.} Gaikwad et al.~\cite{gaikwad_heublein} combine episodic prototypical learning with FedAvg for GNSS interference classification under non-i.i.d.\ data. We build on this learning setup to address model ownership and tracing. The evaluation studies watermarking on an episodic prototypical model.

%% file: 03_preliminaries.tex
\section{PRELIMINARIES}
\label{label_preliminaries}

\textbf{Notation.} For binary vectors $\mathbf{x},\mathbf{y}\in\{0,1\}^m$, $\lVert\mathbf{x}\rVert_1$ denotes Hamming weight and $\mathbf{x}\oplus\mathbf{y}$ denotes bitwise XOR. Their Hamming distance is $\lVert\mathbf{x}\oplus\mathbf{y}\rVert_1$. We write $\mathbf{x}\,||\,\mathbf{y}$ for concatenation and $x\xleftarrow{R}X$ for a uniform draw. The watermark bit subset $\mathcal{B}\subseteq\{1,\dots,n\}$ differs from the Bernoulli noise law $\mathcal{B}_\tau$. Table~\ref{tab:notation} summarizes the symbols. Section~\ref{label_method_system_threat} defines the threat model.

\renewcommand{\arraystretch}{1.02}
\begin{table}[t]
\setlength{\tabcolsep}{4pt}
    \centering
    \caption{Principal notation and deployed dimensions. Identity length $n$ and tracing length $n_{\mathrm{t}}$ describe separate watermark layers.}
    \label{tab:notation}
    \footnotesize
    \resizebox{\columnwidth}{!}{%
    \begin{tabular}{ll}
    \toprule
    \textbf{Symbol} & \textbf{Meaning (space / value)} \\
    \midrule
    $N_c$ & number of clients ($=10$) \\
    $c_i,\ \mathcal{D}_i$ & client $i$ and its private dataset \\
    $\mathbf{W},\mathbf{W}_i$ & global / local model parameters \\
    $\lambda_k$ & FedAvg weights ($\textstyle\sum_k\lambda_k{=}1$) \\
    $\boldsymbol{\gamma},\boldsymbol{\gamma}_{\mathrm{agg}}$ & BN scale carrier ($\in\mathbb{R}^{\omega}$) \\
    $\omega$ & carrier dimension ($=4800$) \\
    $n$ & per-client codeword length ($=128$) \\
    $\mathbf{w}_i,\ \hat{\mathbf{h}}_i$ & codeword / extracted component ($\in\{0,1\}^{n}$) \\
    $E,\ E_i$ & shared / client projection ($\mathbb{R}^{\omega\times N_c n}$, $\mathbb{R}^{\omega\times n}$) \\
    $\rho$ & identity-column load $N_c n/\omega$ \\
    $\mathbf{A}_i,\mathbf{y}_i$ & xLPN public input ($\{0,1\}^{m\times l}$, $\{0,1\}^{m}$) \\
    $\mathbf{s}_i,\mathbf{e}_i$ & xLPN witness ($\{0,1\}^{l}$, $\{0,1\}^{m}$; $\|\mathbf{e}_i\|_1=w_\tau$) \\
    $\mathcal{B}_\tau,\ \tau$ & i.i.d.\ Bernoulli noise, rate $\tau$ \\
    $\mathcal{B}$ & per-episode WM bit subset ($\subseteq\{1,\dots,n\}$) \\
    $\mathsf{err}_n,\ p_r$ & near-collision radius, detection threshold \\
    $f_\theta,\ d$ & embedding net, feature dim ($\mathbb{R}^{d}$, $d{=}512$) \\
    $N,K,Q$ & episodic way / shot / query \\
    $\mathbf{c}_k$ & class prototype ($\in\mathbb{R}^{d}$) \\
    $X_i,\ n_{\mathrm{t}}$ & tracing row / length ($512$ weight, $2048$ feature) \\
    $c,\ k$ & collusion design target / tested coalition size \\
    $S_i,\ Z,\ \varepsilon_1$ & tracing score / threshold / false-accusation target \\
    $q$ & residual flip probability relative to intended output \\
    \bottomrule
    \end{tabular}}
\end{table}

\textbf{Federated learning.} A server coordinates $N_c$ clients, each with private data $\mathcal{D}_i$ and local parameters $\mathbf{W}_i$. FedAvg forms $\mathbf{W}=\sum_{k=1}^{N_c}\lambda_k\mathbf{W}_k$, where $\lambda_k\ge0$ and $\sum_k\lambda_k=1$~\cite{Fedavg}. The same weighted average applies to the BN scales that carry the identity marks. The survival analysis asks whether those marks remain distinguishable after averaging.

\textbf{Prototypical networks.} An $N$-way $K$-shot episode samples $N$ classes, with $K$ labeled support examples and $Q$ query examples per class. Let $S_k$ be class $k$'s support set. The embedding $f_\theta:\mathcal{X}\to\mathbb{R}^d$ forms its prototype as the mean $\mathbf{c}_k=|S_k|^{-1}\sum_{(\mathbf{x}_j,y_j)\in S_k}f_\theta(\mathbf{x}_j)$. A query is classified by a softmax over negative squared Euclidean distances to these prototypes~\cite{snell_swersky}. Training minimizes the query negative log-likelihood over episodes. We use the federated few-shot setup of Gaikwad et al.~\cite{gaikwad_heublein}.

\textbf{Credentials.} In search LPN, the public pair $(\mathbf{A},\mathbf{y})$ satisfies $\mathbf{y}=\mathbf{A}\mathbf{s}\oplus\mathbf{e}$, with noise rate $0<\tau<\tfrac12$. Recovering the secret is an average-case random-code decoding problem~\cite{esser2022syndrome}. The xLPN variant fixes the error weight to $w_\tau=\lfloor m\tau+0.5\rfloor$~\cite{jain_krenn_pietrzak}. Client $i$ publishes $(\mathbf{A}_i,\mathbf{y}_i)$ and retains the witness $(\mathbf{s}_i,\mathbf{e}_i)$.

Verification establishes knowledge of this witness using a three-move commitment--challenge--response $\Sigma$-protocol~\cite{damgard}. Its non-interactive form binds the challenges to the extracted model component and the registered credential (Sec.~\ref{label_method_verification}). Appendix~\ref{app:prop-security} specifies the commitment and random-oracle assumptions.

%% file: 04_method.tex
\section{ZK-TRACE CONSTRUCTION}
\label{label_method}

\begin{figure*}[!t]
    \centering
    \includegraphics[width=\textwidth]{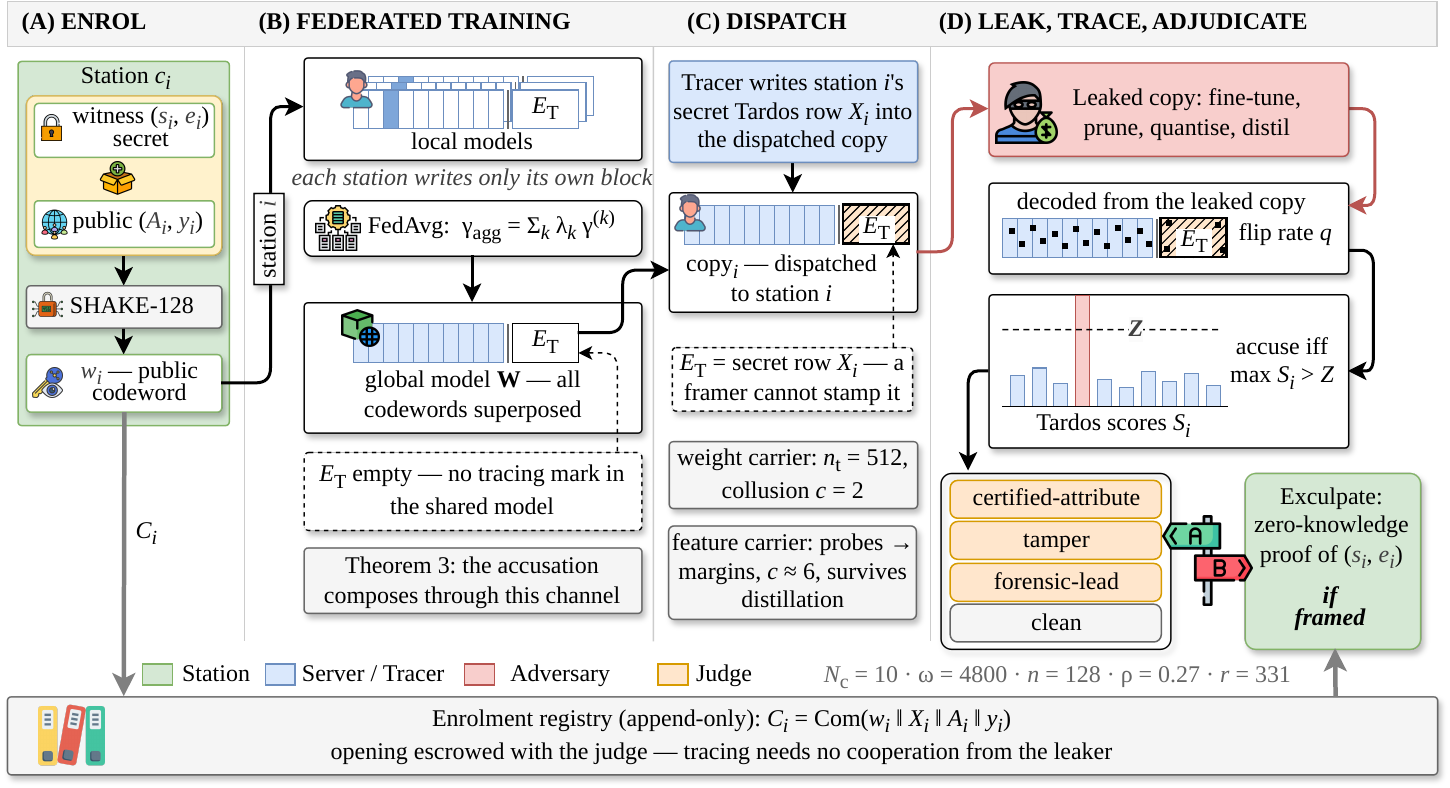}
    \caption{ZK-Trace life cycle and carrier layout. The bar is the projection codebook read from
    the batch-normalization scales $\boldsymbol{\gamma}$, partitioned into per-client identity
    blocks $E_1,\dots,E_{N_c}$ and the tracing overlay block $E_T$. Solid cells carry the public
    codeword $\mathbf{w}_i$, hatched cells the tracer-secret Tardos row $X_i$. $S_i$ is the
    \v{S}kori\'{c}-symmetric score, $Z$ the accusation threshold, and $q$ the per-bit flip rate
    induced by aggregation.}
    \label{figure_method_overview}
\end{figure*}

\textbf{Pipeline.} ZK-Trace has two watermark layers (Fig.~\ref{figure_method_overview}). The identity layer extends FedZKP~\cite{yang_yin_zhu}: each client's xLPN public input determines a codeword embedded in the shared model's BN scales. The tracing layer adds a secret Tardos fingerprint only to the copy dispatched to that client. It supports attribution after modifications invalidate an exact dispatch-log hash, providing evidence for trace-and-revoke~\cite{naor2001revocation}.

Fig.~\ref{figure_method_overview} follows one station through four stages. At~(A), station $c_i$ registers $(\mathbf{A}_i,\mathbf{y}_i)$ and its derived codeword $\mathbf{w}_i$, keeping the witness $(\mathbf{s}_i,\mathbf{e}_i)$ private. At~(B), federated training aggregates the identity codewords into $\mathbf{W}$ while leaving the tracing block $E_T$ empty. At~(C), the tracer embeds the secret row $X_i$ into the copy issued to station~$i$. At~(D), the tracer extracts the leaked copy's overlay, scores each registered row, and accuses a station when its score exceeds $Z$.

\textbf{Binding intuition.} Each client has an assigned set of projection directions at which to read its public codeword. Agreement with that codeword is statistical evidence of a mark. Possession of its private xLPN witness is a separate credential fact. A valid proof authenticates the witness holder. The verifier also includes a digest of the entire model state in its context, so equal extracted bits do not permit replay onto different model bytes. Neither presence nor authentication establishes who trained the model: public marks can be copied. Recipient tracing therefore uses the separate secret row and registry.

\subsection{Setting and Threat Model}
\label{label_method_system_threat}

The server acts as tracer and issues each authorized client a copy bearing its secret Tardos row. The registry allows a recovered copy to be traced without contacting its holder. Clients keep their witnesses private, and the tracer holds the tracing rows.

\textbf{Threat model.} We consider authorized clients that hold valid credentials but may misuse their dispatched copies. Access control for unauthorized parties is a separate concern.

\setlength{\leftmargini}{14pt}
\begin{itemize}
\setlength\itemsep{2pt}
\item \textbf{T1, weight-space post-processor.} A leaker applies utility-preserving edits such as fine-tuning, pruning, quantization, or noise before redistributing its copy. We evaluate both identity attribution and Tardos tracing under these edits. Their guarantees require the extraction conditions stated in Sec.~\ref{sec:theory-survival}.

\item \textbf{T2, bounded coalition.} Clients may average their copies to attenuate the marks~\cite{boneh1998collusion}. The design targets are two colluders on the weight carrier and three on the feature carrier. Soundness requires innocent-row independence, A5(a). Finite channel completeness additionally requires A5(b). A coalition may base its intended symbols on its entire row matrix under that theorem. The experiments distinguish model averaging from decoder-space adaptive steering.

\item \textbf{T3, framing adversary.} An adversary may copy a victim's public identity mark. It lacks the victim's witness and secret tracing row. Certified false naming is bounded under A5(a), while credential verification authenticates the claimant. Authentication alone does not exculpate a recipient: the judge must examine the accusation evidence.

\item \textbf{T4, distiller.} A leaker may train a fresh student whose BN parameters contain no identity codeword. This erases the weight carrier in our benchmark~\cite{lukas2022sok}. The feature carrier can survive when the student copies the teacher's feature geometry closely enough (Proposition~\ref{prop:kd-robust}). Function-only distillation falls outside Assumption~\ref{ass:kd} and erases this carrier too.
\end{itemize}

\textbf{Out of scope.} ZK-Trace does not detect or prevent poisoned updates. It can be combined with robust aggregation to address this threat (Sec.~\ref{sec:eval_robustness}). Credential redistribution is also outside its scope. The fingerprint identifies the enrolled recipient of the issued copy, and subsequent revocation is handled through the registry~\cite{naor2001revocation}.

\textbf{Server as leaker.} The server and tracer are honest-but-curious. The server holds the shared model and issued copies, so the scheme cannot distinguish a server leak from a recipient leak of the same copy. Registry commitments prevent post-hoc row substitution. Authenticated delivery and server accountability require additional evidence (App.~\ref{label_method_registry}).

\subsection{Codeword Embedding}
\label{label_method_watermark_embedding}

Before training, initialization fixes the shared Gaussian matrix $E\in\mathbb{R}^{\omega\times N_c n}$, with $\omega{=}4{,}800$ and $n{=}128$, and the detection threshold $p_r$ and radius $\mathsf{err}_n$. Following FedZKP~\cite{yang_yin_zhu}, we select an xLPN instance with a syndrome-decoding work estimate reported near $2^{152}$~\cite{esser2022syndrome} as a separate computational-hardness estimate (Sec.~\ref{label_method_verification}; dimensions in App.~\ref{app:method-detail}). Each client's public input determines its identity codeword,
\begin{equation}
    \mathbf{w}_i = \text{SHAKE-128}(\mathbf{A}_i \,|\, \mathbf{y}_i) \in \{0,1\}^{n}.
\label{eq_watermark_hash}
\end{equation}
Client $i$ uses block $E_i\in\mathbb{R}^{\omega\times n}$, comprising columns $[(i{-}1)n,in)$ of $E$. The server retains the concatenated public inputs $(\mathbf{A}_{\mathrm{agg}},\mathbf{y}_{\mathrm{agg}})$ so a verifier can recompute any station's codeword.

\textbf{Carrier.} The vector $\boldsymbol{\gamma}\in\mathbb{R}^{\omega}$ concatenates the scales of all $L{=}20$ BN layers. Each projection spreads a bit across these scales, so editing one layer affects only part of its carrier. At $N_c{=}10$, the identity-codebook load is $N_c n/\omega\approx0.27$. BN scales rescale activations after normalization. Section~\ref{sec:eval_robustness} evaluates how this affects feature geometry and classification.

\textbf{Projection and objective.} For bit $b$, the projection $z_{i,b} = \boldsymbol{\gamma}^\top E_{i,b}$ gives the extracted bit $\hat{h}_{i,b} = \mathbb{1}[z_{i,b} > 0]$, with signed target $t_{i,b} = 2w_{i,b} - 1$. The hinge loss penalizes a wrong sign or a signed projection smaller than the target margin $\mu>0$,
\begin{equation}
    \mathcal{L}_{\mathrm{wm}} = \frac{1}{|\mathcal{B}|} \sum_{b \in \mathcal{B}} \max\!\big(0,\; \mu - t_{i,b} \cdot z_{i,b}\big),
    \label{eq_watermark_loss}
\end{equation}
with the active subset $\mathcal{B} \subset \{1, \ldots, n\}$ drawn each step by bit-dropout ($p_{\mathrm{drop}} = 0.5$, $\mu = 1.0$). The local objective adds this to the prototypical task loss, $\mathcal{L} = \mathcal{L}_{\mathrm{task}} + \lambda_{\mathrm{wm}} \mathcal{L}_{\mathrm{wm}}$, with $\lambda_{\mathrm{wm}} = 0.1$ trading accuracy against watermark strength.

\textbf{Federated training.} Cross-entropy pre-training on base classes produces an unwatermarked initial model~\cite{gaikwad_heublein}. Federated episodic training then follows Sec.~\ref{label_preliminaries}, with $\lambda_k$ proportional to local sample count. Each client fixes its public input and codeword in the first round and trains under $\mathcal{L}$ in subsequent rounds. Projection onto $E_i$ recovers client $i$'s component from the averaged BN scales, subject to cross-talk from the other clients (Sec.~\ref{sec:theory-survival}).

\subsection{Zero-Knowledge Verification}
\label{label_method_verification}

Algorithm~\ref{alg_verify} verifies credential knowledge using a Stern-type non-interactive proof~\cite{yang_yin_zhu}. All commitments and the public context determine the challenge vector. In the verifier, $\mathsf{ctx}=\mathsf{tag}\|H_{256}(\mathbf W)\|\hat h_i\|A_i\|y_i$, where $H_{256}$ hashes a canonical full state including buffers. We use $r=331$ repetitions. Appendix~\ref{app:prop-security} separates the grinding calibration from end-to-end witness-recovery security.

\begin{algorithm}[!t]
\caption{Non-interactive credential and model-presence verification, from FedZKP's Stern-type xLPN $\Sigma$-protocol~\cite{yang_yin_zhu}.}
\label{alg_verify}
\footnotesize
\begin{algorithmic}[1]
\REQUIRE witness $(\mathbf{s}_i,\mathbf{e}_i)$, $\mathbf{y}_i{=}\mathbf{A}_i\mathbf{s}_i{\oplus}\mathbf{e}_i$; block $(\mathbf{A}_i,\mathbf{y}_i)$; model $\mathbf{W}$; rounds $r{=}331$
\ENSURE verify credential knowledge and mark presence, else reject
\STATE $\hat{\mathbf{h}}_i \gets \mathbb{1}[\boldsymbol{\gamma}(\mathbf{W})^{\!\top}E_i > 0]$,\ \ $\mathsf{ctx} \gets \mathsf{tag}\|H_{256}(\mathbf W)\|\hat{\mathbf{h}}_i\|\mathbf A_i\|\mathbf y_i$
\STATE \textbf{Prover} (offline), for $k{=}1,\dots,r$ with fresh $\pi,\mathbf{v},\mathbf{f}$:
\STATE \quad $t_0{\gets}\mathbf{A}_i\mathbf{v}{\oplus}\mathbf{f}$, $t_1{\gets}\pi(\mathbf{f})$, $t_2{\gets}\pi(\mathbf{f}{\oplus}\mathbf{e}_i)$
\STATE \quad $C_0{\gets}\mathrm{com}(\pi,t_0)$, $C_1{\gets}\mathrm{com}(t_1)$, $C_2{\gets}\mathrm{com}(t_2)$
\STATE $(c^{(k)})_{k=1}^{r} \gets \mathrm{SHAKE\text{-}128}\big(C_0^{(1)}\|\cdots\|C_2^{(r)}\,\|\,\mathsf{ctx}\big) \bmod 3$
\STATE open the two commitments per $c^{(k)}$ and send transcript $\tau$
\STATE \textbf{Verifier} (offline). Recompute $(c^{(k)})$ from $\tau,\mathsf{ctx}$ and check per round:
\STATE \quad $c^{(k)}{=}0$:\ \ $t_0 \oplus \pi^{-1}(t_1) \in \mathrm{img}(\mathbf{A}_i)$
\STATE \quad $c^{(k)}{=}1$:\ \ $t_0 \oplus \pi^{-1}(t_2) \oplus \mathbf{y}_i \in \mathrm{img}(\mathbf{A}_i)$
\STATE \quad $c^{(k)}{=}2$:\ \ $\mathrm{wt}(t_1 \oplus t_2) = w_\tau$
\STATE \textbf{accept} iff all $r$ checks pass and $\mathrm{HD}(\hat{\mathbf{h}}_i,\mathbf{w}_i) \le \mathrm{err}_n$, else \textbf{reject}
\end{algorithmic}
\end{algorithm}

\textbf{Attribution and acceptance.} From the public block $E_i$ and codeword $\mathbf{w}_i$ alone, a verifier extracts $\hat{\mathbf{h}}_i = \mathbb{1}\!\big[\boldsymbol{\gamma}(\mathbf{W})^\top E_i > 0\big]$ and tests codeword presence,
\begin{equation}
    \mathrm{HD}\big(\hat{\mathbf{h}}_i, \mathbf{w}_i\big) \leq \mathsf{err}_n.
    \label{eq_acceptance}
\end{equation}
Against an independent codeword at $n{=}128$, calibration to $2^{-128}$ requires $\mathsf{err}_n{=}0$. Exact presence is restrictive for the measured noisy extraction, so the deployed identity metric uses maximum codeword agreement (Prop.~\ref{prop:presence-honest}, App.~\ref{app:proofs}). A match is statistical evidence of a codeword, but it does not authenticate its holder because the public mark can be copied. Algorithm~\ref{alg_verify} adds proof of the private witness and retains the explicit presence condition in its acceptance rule. Recipient tracing instead uses the secret overlay below.

\subsection{Security Guarantees}
\label{label_method_security}

Here $Q$ denotes the number of random-oracle queries, distinct from query examples per episode.

\begin{proposition}[Credential knowledge and artifact binding]
\label{prop:security}
Under the random-oracle and commitment assumptions in Appendix~\ref{app:prop-security}, the ideal Stern/Fiat--Shamir protocol is complete for a witness holder whose codeword passes the requested presence test, and admits a zero-knowledge simulation. With ideal binding commitments and uniform challenges, its knowledge error is at most $(Q+1)(2/3)^r$~\cite{attema2022fiatshamir}. The artifact-bound context ties a transcript to the registered credential, extracted component, and full-state digest. Transfer to a different state requires a digest collision or a new-context challenge coincidence. These are knowledge and binding properties, not a proof of authorship or a numerical bound on all xLPN-recovery attacks.
\end{proposition}

\noindent The proof, with extractor, simulator, and grinding accounting, is in Appendix~\ref{app:proofs}.

\textbf{Two soundness layers.} Credential verification and recipient tracing answer different questions. The former proves knowledge of a private witness and binds the statement to an artifact. The latter bounds naming an innocent recipient under a conditional code model. The $331$-round grinding calibration, commitment failures, computational witness recovery, and statistical tracing budget must be accounted for separately (Table~\ref{tab:soundness-split}).

\subsection{Tracing Overlay and Two Carriers}
\label{label_method_tracing}

\textbf{Dispatch and readout.} The identity layer cannot isolate a recipient because the aggregate contains every contributor's mark. At dispatch, the tracer fine-tunes a copy on server-held proxy data, combining the task loss with a hinge that embeds the recipient's secret Tardos row $X_i$. The weight carrier reads each bit from a BN-scale projection. The feature carrier reads the sign of a feature projection on a fixed probe input, with co-training used to align those signs with $X_i$. Its design target is $c{=}3$, compared with $c{=}2$ for the weight carrier. Feature-matching distillation can preserve these probe responses under Assumption~\ref{ass:kd}.

\textbf{Tracing and adjudication.} The tracer scores the recovered bits against registered rows. The adjudication rule certifies each positive or negative tail before issuing a certified decision, allocating one budget across both tails and any jointly used carriers (Prop.~\ref{prop:adjudicate}). A threshold exceedance without a passing certificate remains an uncertified lead. If neither threshold is exceeded, the outcome is no certified evidence. Threshold-only decisions are evaluated separately to measure score separation. An independent judge checks the enrollment opening and reconstructs the score evidence. A credential proof authenticates the claimant but does not decide whether that claimant leaked.

%% file: 05_theory.tex
\section{THEORETICAL ANALYSIS}
\label{sec:theory-survival}

The analysis separates three questions: when identity bits survive averaging, when a tracing score supports a bounded false-accusation probability, and when feature matching preserves probe bits. The identity probability model is idealized. The deterministic decoding and probe-margin results do not require that model. Appendix~\ref{sub:assumptions} states the assumptions and Appendix~\ref{app:proofs} gives the proofs.

\textbf{Identity survival.} Write $\gamma^{(i)}=\gzero+\Delta_i$, $G=\|\gzero\|$, and $u_i=\sum_{k\ne i}\lambda_k\gamma^{(k)}$. A local signed margin of at least $\mu$ gives
\begin{equation}
\label{eq:identity-proxy-main}
t_{i,b}\inner{\gagg}{E_{i,b}}\ge\lambda_i\mu+t_{i,b}\inner{u_i}{E_{i,b}}.
\end{equation}
A probabilistic recovery bound follows from this inequality when the signed directions are independent of the cross-vector.

\begin{theorem}[Conditional per-bit recovery]
\label{thm:perbit}
Assume A1--A3, $\omega>4$, and $\lambda_i>0$. Conditional on $(u_i,t_i)$, let $F_\omega$ be the CDF of $\sqrt\omega$ times the first coordinate of a uniform unit vector. For $u_i\ne0$, put $s_i=\lambda_i\mu\sqrt\omega/\|u_i\|$. Then
\begin{equation}
\label{eq:perbit-gauss}
\begin{split}
\Prob[\hat h_{i,b}=w_{i,b}\mid u_i,t_i]&\ge F_\omega(s_i),\\
F_\omega(s_i)&\ge\Phi(s_i)-b_\omega,\qquad b_\omega=\frac{8}{\omega-4}.
\end{split}
\end{equation}
The proxy events in \eqref{eq:identity-proxy-main} are conditionally independent across bits. If $u_i=0$, every bit is correct.
\end{theorem}

A deterministic bound $\|\Delta_k\|\le D_k$ gives $\|u_i\|\le(1-\lambda_i)G+\sum_{k\ne i}\lambda_kD_k$ by the triangle inequality. Random projection columns alone do not establish the smaller quadrature norm. For interpreting the measurements we retain the approximation
\begin{equation}
\label{eq:snr-scaling}
\SNR_{\mathrm{quad}}:=\frac{\mu\sqrt\omega}{\sqrt{(N_c-1)[(N_c-1)G^2+n\mu^2]}}.
\end{equation}
It assumes uniform weights, update norms near $\mu\sqrt n$, and negligible cross terms. It is a diagnostic model, not a proved floor for trained federated networks.

\textbf{Attribution by codeword separation.} Maximum agreement is nearest-neighbor decoding in Hamming distance. A rival participates in training, so its codeword need not be independent of the decoded aggregate. The following criterion avoids that independence assumption.

\begin{theorem}[Attribution from a decoding radius]
\label{thm:attr}
Let $d_i=\HD(\hat h_i,w_i)$ and $d_{\min}=\min_{i\ne j}\HD(w_i,w_j)$. If $2d_i<d_{\min}$ for every client, all maximum-agreement decisions are uniquely correct. Under A4, if $\Prob[\max_i d_i>rn]\le\eta$ for some $0\le r<1/4$, then
\begin{equation}
\label{eq:attr-thresh}
\Prob[\text{any attribution error}]
\le\eta+\binom{N_c}{2}e^{-2n(1/2-2r)^2}.
\end{equation}
The radius event may depend on the entire codebook.
\end{theorem}

\begin{corollary}[A sufficient operating condition]
\label{cor:envelope}
Under A1--A4, suppose $F_\omega(s_i)\ge p$ uniformly over clients and conditioning values. Choose $\xi>0$ with $r=1-p+\xi<1/4$. All clients are attributed correctly with probability at least $1-\delta$ if
\begin{equation}
\label{eq:master}
N_c e^{-2n\xi^2}+\binom{N_c}{2}e^{-2n(1/2-2r)^2}\le\delta.
\end{equation}
\end{corollary}

These are sufficient conditions, not necessary thresholds. A measured mean bit-accuracy cannot substitute for the uniform conditional $p$. The load $\rho=N_cn/\omega$ describes the number of identity constraints relative to scale dimension. $\rho=1$ is a rank boundary, not a universal attribution-failure point (Lemma~\ref{lem:capacity}). Partial participation changes the current averaging weights. Marks from earlier rounds can remain in a skipped client's absence. Presence additionally requires a specified absolute Hamming radius, calibrated in Proposition~\ref{prop:presence-honest}.

\textbf{Tracing soundness.} For a recovered tracing word $y$, biases $p_b$, and score $S_i=\sum_b U(X_{i,b},y_b,p_b)$, an innocent row has independent Bernoulli entries conditional on $(y,p)$ under A5(a). Its first two moments are unchanged by the output, although its full distribution depends on the output.

\begin{theorem}[Conditional tracing bounds]
\label{thm:tardos}
Assume A5(a) and Definition~\ref{def:overlay}. For any threshold $z>0$ and $\alpha>0$, define
\begin{equation}
\label{eq:mgf-chernoff}
\begin{split}
M_b(\alpha;y,p)&=(1-p_b)e^{-\alpha(2y_b-1)\sqrt{p_b/(1-p_b)}}\\
&\quad+p_b e^{\alpha(2y_b-1)\sqrt{(1-p_b)/p_b}},\\
E_\alpha(y,p;z)&=\alpha z-\sum_b\log M_b(\alpha;y,p).
\end{split}
\end{equation}
Then
\begin{equation}
\label{eq:tardos-sound}
\Prob[\exists\text{ innocent }i:S_i>z\mid y,p]
\le\min\{1,N e^{-E_\alpha(y,p;z)}\}.
\end{equation}
A separate a-priori bound holds uniformly over $(y,p)$ at
\begin{equation}
\label{eq:tardos-sound-proved}
z_B=\frac{BL}{3}+\sqrt{\left(\frac{BL}{3}\right)^2+2\nt L},
\end{equation}
where $B=\sqrt{(1-\delta_c)/\delta_c}$, $\delta_c=1/(300c)$, and $L=\ln(N/\varepsilon_1)$: the false-accusation probability is at most $\varepsilon_1$.
\end{theorem}

The candidate threshold $Z=\sqrt{2\nt L}$ used in part of the evaluation is smaller than $z_B$. It supports a certificate only when an evaluated exponent passes the required budget. The coalition sweep measures threshold exceedances, while the certificate-based evaluation checks the probability bound for each decoded word (Table~\ref{tab:certified-audit}). Appendix~\ref{ssub:adjudicate} specifies a certificate-gated rule, including separate budgets for positive and negative decisions.

\textbf{Finite completeness.} Theorem~\ref{thm:finite-completeness} supplies a lower-tail guarantee for the actual cutoff, without assuming independent intended coalition symbols. For a coalition of size $s\le c$, its exponential-moment bound is
\begin{equation}
\label{eq:finite-main}
\Prob[\max_{i\in\Coll}S_i\le z]\le e^{tsz}J_s(t,q)^{\nt},\qquad t>0.
\end{equation}
The function $J_s$ integrates over the posterior secret bias given each coalition column, maximizing over permitted output bits. Conditioning on the entire coalition matrix allows the attacker to choose its intended word jointly across positions. Only the residual channel must flip these bits independently with probability $q$. Appendix~\ref{app:tardos-proofs} proves the bound and gives a compatible a-priori soundness threshold.

Interval integration establishes both guarantees at the two deployed lengths: for $N=10$ and false-accusation budget $10^{-3}$, the worst missed-coalition bounds over $s\le c$ are $0.0096$ at $(c,\nt,q)=(2,512,0.04)$ and $0.00201$ at $(3,2048,0.15)$ (Table~\ref{tab:finite-guarantees}). These are finite channel guarantees, not estimates of the networks' residual channel. Per-colluder disagreement after averaging cannot identify $q$.

For planning, the asymptotic length estimate is
\begin{equation}
\label{eq:tardos-length}
\nt^{\mathrm{design}}=\left\lceil\frac{\pi^2}{2}c^2(1-2q)^{-2}\ln(N/\varepsilon_1)\right\rceil.
\end{equation}
Its inverse-square channel cost~\cite{skoric2008symmetric,laarhoven2014optimal} describes a design trend. Equation~\eqref{eq:finite-main} establishes the finite guarantee for the specified threshold and cutoff. Figure~\ref{fig:channel-price} distinguishes that trend from threshold certification.

\begin{figure}[t]
\centering
\includegraphics[width=\columnwidth]{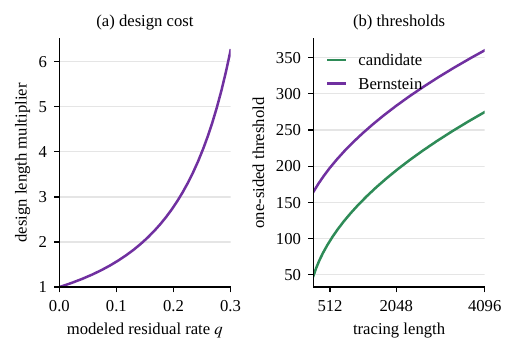}
\caption{Tracing design and certification are separate. (a)~Inverse-square length multiplier for a hypothetical independent residual flip channel. (b)~Candidate and Bernstein-certified one-sided thresholds versus tracing length, at $c=2$, $N=10$, and $\varepsilon_1=10^{-3}$. The certified threshold does not assert completeness.}
\label{fig:channel-price}
\end{figure}

\textbf{Feature stability under distillation.} Let $f_T,f_S$ use the same feature coordinates and let $\eKD^2=\nt^{-1}\sum_b\|f_S(p_b)-f_T(p_b)\|_2^2$ on the watermark probes.

\begin{proposition}[Empirical probe-margin stability]
\label{prop:kd-robust}
For the unit projections in Definition~\ref{def:feature-carrier}, suppose all but a fraction $q_{\mathrm{bad}}$ of teacher probes have the correct signed margin at least $\mu>0$. The student's disagreement with the embedded binary row satisfies
\begin{equation}
\label{eq:qkd}
q_{\mathrm{KD}}\le\min\left\{1,q_{\mathrm{bad}}+\frac{\eKD^2}{\mu^2}\right\}.
\end{equation}
\end{proposition}

This bound counts wrong and insufficient-margin teacher probes in $q_{\mathrm{bad}}$. A stronger score-level result handles arbitrary correlated errors: if at most $k$ decoded bits change, then $S_i^S\ge S_i^T-A_i(k)$, where $A_i(k)$ sums the $k$ largest weights $2|X_{i,b}-p_b|/\sqrt{p_b(1-p_b)}$. Thus $S_i^T-A_i(k)>z$ certifies threshold survival without a binary-symmetric channel (Theorem~\ref{thm:score-stability}). Feature matching can preserve probe margins. Function-only matching, however, can change feature coordinates and violate the margin condition. Weight projections are not fixed by a feature-matching objective. Section~\ref{sec:eval_robustness} reports the observed outcomes for both carriers.

%% file: 06_experimental_setup.tex
\section{EXPERIMENTAL SETUP}
\label{label_experiments}

\textbf{Datasets.} We re-render the GNSS recordings of Heublein et al.~\cite{heublein_feigl_crpa} as $100{,}000$ four-channel $4{\times}32{\times}32$ antenna-array tensors. The GNSS benchmark contains six interference classes and excludes the interference-free class. We use a sample-level $80/10/10$ split (not a held-out-class split) into training, validation, and test sets. Class imbalance reaches approximately $178{\times}$. A Dirichlet partition with $\alpha{=}2.0$ gives approximately $8{,}000$ samples per client and moderate label skew, without guaranteeing class coverage. CIFAR-10 ($3{\times}32{\times}32$) provides the transfer check.

All ten stations use partitions of one measurement campaign. These partitions vary class availability and sample count, but do not reproduce separately sited receivers with distinct calibration, multipath, and local interference. The receiver-shift experiment is a proxy for these effects. Appendix~\ref{sub:config} discusses the unmodeled variation and the need to recheck tracing certificates on decodes from new sites.

\textbf{Implementation.} A ResNet-18 backbone~\cite{he_zhang}, with its final fully connected layer removed, feeds a prototypical head~\cite{snell_swersky}. We pre-train each dataset for $200$ cross-entropy epochs, then run $100$ FedAvg rounds~\cite{Fedavg} with ten clients. Each client trains on $25$ local five-way five-shot episodes per round, with $15$ query examples per class. Both phases use stochastic gradient descent (SGD), with optimizer settings and seeds listed in Appendix~\ref{sub:config}.

\textbf{Watermark configuration.} Each client has $n{=}128$ identity bits. We evaluate the deployed $N_c{=}10$ setting ($\rho{\approx}0.27$) and a higher-load setting with $N_c{=}40$ ($\rho{=}1.07$). The presence calibration uses $p_r{=}2^{-128}$, giving $\mathsf{err}_n{=}0$ (Prop.~\ref{prop:presence-honest}). The identity and presence criteria are distinguished below. Credential proofs use $r{=}331$ parallel repetitions. On dispatched copies, the weight overlay uses $n_{\mathrm{t}}{=}512$ BN-$\gamma$ projections with design target $c{=}2$. The feature carrier uses $n_{\mathrm{t}}{=}2048$ in-distribution probes with design target $c{=}3$.

\textbf{Metrics.} We measure classification accuracy over $200$ five-way five-shot test episodes. Each GNSS episode samples five of the six classes. On the aggregate, \emph{attribution accuracy} is the fraction of clients identified by maximum codeword agreement. \emph{Self bit-accuracy} measures agreement with the client's own codeword, while \emph{cross-talk} measures mean agreement with other codewords. \emph{Presence} is the stricter Hamming-distance test in \eqref{eq_acceptance}.

On dispatched copies, \emph{single-leaker isolation} requires naming the owner without accusing an innocent. Testing every recipient over ten seeds gives $100$ trials at $N_c{=}10$ and $400$ at $N_c{=}40$, per dataset. \emph{Collusion traceability} requires accusing at least one colluder and no innocent, for sampled coalition sizes $k\in\{2,3,5,8\}$. We also test offline tracing and resistance to codeword copying and evidence substitution (Sec.~\ref{label_method_tracing}). Distillation counts pool ten runs per dataset, giving twenty runs per distiller class. Coalition sampling counts are in Appendix~\ref{sub:config}. The certificate-based evaluation uses $N_c=10$ weight-carrier models from ten CIFAR-10 and six GNSS seeds. For each seed, it tests all ten single copies and all 45 pair averages with the adjudication rule in Proposition~\ref{prop:adjudicate}.

Feature geometry is measured by silhouette score, intra- and inter-class distance, and leave-one-out $1$-NN accuracy. We report ten-seed means and standard deviations or pooled counts, as indicated. Matched watermark-off/on models differ only in $\lambda_{\mathrm{wm}}$. Two one-sided tests (TOST) assess global-model equivalence at the stated margins ($0.5$\,pp on CIFAR-10 and $1.0$\,pp on GNSS, App.~\ref{sub:utility-ext}).

\textbf{Attacks.} Each attack is evaluated on $100$ few-shot test episodes. We report task accuracy together with presence and attribution, so successful erasure can be distinguished from destruction of the classifier. Model-modification attacks include BN-scale pruning, Gaussian noise, quantization, and combined pruning and quantization. Targeted attacks use projected gradient descent (PGD) against a bit-flip objective or reset selected BN layers. Training-based attacks use fine-tuning without the watermark loss or knowledge distillation (KD)~\cite{hinton_vinyals_dean}. The latter includes feature matching~\cite{romero2015fitnets}, cross-architecture transfer, logit-only distillation, and feature isometry.

Additional tests cover insider own-row erasure (negation, fresh re-randomization, and their per-carrier hybrid), robust aggregation~\cite{blanchard2017byzantine} using the coordinate-wise median~\cite{yin2018byzantine}, and GNSS receiver covariate shift. Appendix~\ref{sub:attack-protocol} gives the grids and budgets.

\textbf{Baselines.} Five federated-watermarking methods use the same learning pipeline: FedZKP~\cite{yang_yin_zhu}, FedTracker~\cite{shao_yang_fedtracker}, DUW~\cite{yu_hong_duw}, FedIPR~\cite{li_fan_gu}, and WAFFLE~\cite{tekgul_xia_marchal}. FedZKP supplies the group-ownership reference, while FedTracker and DUW supply per-client fingerprints. FedIPR uses BN scales, while WAFFLE uses its native trigger carrier. Each method is evaluated under a shared seven-attack suite using its own mark metric. We additionally test a DeepMarks-style balanced incomplete block design (BIBD) comparator~\cite{chen_deepmarks} under the coalition sweep, because the five federated baselines lack collusion-secure codes.

%% file: 07_results.tex
\section{RESULTS}
\label{label_evaluation}

We evaluate identity attribution on the aggregate, recipient tracing on dispatched copies, and robustness after attacks. The results distinguish the two carriers' tracing performance and copy-accuracy costs from the utility of the deployed global model.

% ======================================================================
\subsection{Attribution and Survival}
\label{sec:eval_attribution}

\renewcommand{\arraystretch}{0.82}
\begin{table}[t]
    \centering
    \caption{Per-client attribution on the aggregated model ($N_c{=}10$, $n{=}128$, $\rho{=}0.27$; 10-seed mean$\pm$std). Attr.: clients named correctly. Self, Cross: bit-agreement with own and with other codewords. Acc: few-shot accuracy.}
    \label{tab:attribution}
    \resizebox{\columnwidth}{!}{%
    \begin{tabular}{llcccc}
    \toprule
    \textbf{Dataset} & \textbf{Method} & \textbf{Attr.\,(\%)} & \textbf{Self} & \textbf{Cross} & \textbf{Acc\,(\%)} \\
    \midrule
    \multirow{2}{*}{GNSS} & Ours & \textbf{\textcolor{ForestGreen}{$98.0_{\pm4.0}$}} & $0.96_{\pm.02}$ & $0.50_{\pm.01}$ & $93.4_{\pm1.0}$ \\
    & FedZKP (group) & $10.0^{\dagger}$ & -- & -- & $93.9_{\pm0.4}$ \\
    \midrule
    \multirow{2}{*}{CIFAR-10} & Ours & \textbf{\textcolor{ForestGreen}{$100.0_{\pm0.0}$}} & $0.94_{\pm.01}$ & $0.50_{\pm.01}$ & $84.9_{\pm0.3}$ \\
    & FedZKP (group) & $10.0^{\dagger}$ & -- & -- & $84.7_{\pm0.2}$ \\
    \bottomrule
    \end{tabular}}

    \vspace{2pt}
    {\footnotesize $^{\dagger}$~Chance level $1/N_c$ (group mark, no per-client component).}
\end{table}

The maximum-agreement rule identifies $98\%$ of contributors on GNSS and $100\%$ on CIFAR-10 (Table~\ref{tab:attribution}). Mean cross-talk is $0.50$, consistent with the analysis of non-matching codewords. FedZKP's group mark contains no per-client identifier, so its reference attribution rate is chance ($1/N_c{=}10\%$). These results identify contributors to the shared aggregate. Recipient tracing is evaluated separately in Sec.~\ref{sec:eval_tracing}.

\textbf{Survival under aggregation.} Self-agreement decreases as the load ratio $\rho$ grows, but remains above the quadrature diagnostic $\hat p$ in every cell (Table~\ref{tab:survival}, Fig.~\ref{fig:survival}). CIFAR-10 attribution remains exact through $\rho{=}1.07$ ($N_c{=}40$). GNSS attribution falls to $0.87$ there despite similar self-agreement ($0.77$ versus $0.76$). Thus the per-bit floor alone does not explain the dataset difference. In particular, GNSS's smaller baseline scale $G$ raises the quadrature diagnostic and cannot explain its earlier attribution decline. Mean cross-talk remains near $0.50$, but this average does not determine the largest competing score. Appendix~\ref{sub:tightness} discusses the bound's limits near capacity.

\begin{table}[t]
\centering
\footnotesize
\setlength{\tabcolsep}{4pt}
\renewcommand{\arraystretch}{1.05}
\caption{Predicted survival vs.\ ten-seed measurements under FedAvg ($R{=}100$, $n{=}128$, $\mu{=}1$, $\omega{=}4800$; means, with std shown where ${>}0.01$). $\rho{=}N_c n/\omega$: identity load. $\hat p{=}\Phi(\SNR_{\mathrm{quad}})$: quadrature diagnostic \eqref{eq:snr-scaling}. Self, Cross, Attr.: measured per-bit self-agreement, cross-talk, and attribution. Acc: few-shot accuracy. $G{=}\|\gamma_0\|$: backbone baseline scale.}
\label{tab:survival}
\resizebox{\columnwidth}{!}{%
\begin{tabular}{@{}clcccccc@{}}
\toprule[1pt]
\textbf{Dataset} & \textbf{Setting} & $\boldsymbol{\rho}$ & $\boldsymbol{\hat p}$ & \textbf{Self} & \textbf{Cross} & \textbf{Attr.} & \textbf{Acc} \\
\midrule[0.8pt]
\multirow{6}{*}{\rotatebox{90}{\shortstack{\textbf{CIFAR-10}\\($G{=}7.32$)}}} & $N_c{=}5$   & 0.13 & 0.97 & 1.00 & 0.50 & $1.00$            & 0.85 \\
 & $N_c{=}10$  & 0.27 & 0.83 & 0.94 & 0.50 & $1.00$            & 0.85 \\
 & $N_c{=}20$  & 0.53 & 0.68 & 0.86 & 0.50 & $1.00$            & 0.85 \\
 & $N_c{=}40$  & 1.07 & 0.59 & 0.76 & 0.50 & $1.00$            & 0.86 \\
 & $N_c{=}80$  & 2.13 & 0.55 & 0.66 & 0.50 & $0.88_{\pm 0.04}$ & 0.86 \\
 & $N_c{=}120$ & 3.20 & 0.53 & 0.63 & 0.50 & $0.61_{\pm 0.05}$ & 0.86 \\
\midrule
\multirow{4}{*}{\rotatebox{90}{\shortstack{\textbf{GNSS}\\($G{=}4.73$)}}} & $N_c{=}5$   & 0.13 & 0.99 & 1.00 & 0.50 & $1.00$            & 0.92 \\
 & $N_c{=}10$  & 0.27 & 0.90 & 0.96 & 0.50 & $0.98_{\pm 0.04}$ & 0.93 \\
 & $N_c{=}20$  & 0.53 & 0.75 & 0.88 & 0.50 & $0.96_{\pm 0.03}$ & 0.92 \\
 & $N_c{=}40$  & 1.07 & 0.64 & 0.77 & 0.50 & $0.87_{\pm 0.04}$ & 0.86 \\
\bottomrule[0.8pt]
\end{tabular}}
\end{table}

% This figure plots tab:survival and must share its page. Both floats are
% therefore declared HERE, after the paragraph that references them, rather than
% before it: Sec. VII opens low on the preceding page, whose second column is
% already two-thirds tab:attribution, so a third float cannot fit there and
% LaTeX would strand the figure a page behind its table. Declaring the pair
% together after the text sends them to the next page as a unit. Specifier and
% source position both matter -- [!t]/[!tb]/[!ht] alone do not fix it.
\begin{figure}[!t]
    \centering
    \includegraphics[width=\columnwidth]{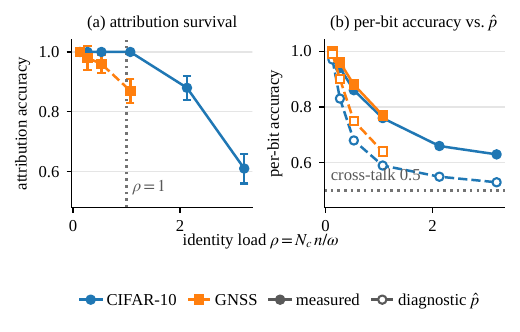}
    \caption{Attribution versus identity load from Table~\ref{tab:survival} ($n{=}128$, $\omega{=}4800$, $R{=}100$, 10-seed). (a)~Attribution accuracy against $\rho{=}N_c n/\omega$, error bars one std. Dotted: rank reference $\rho{=}1$. (b)~Measured per-bit self-agreement (filled, solid) and the diagnostic $\hat p{=}\Phi(\SNR_{\mathrm{quad}})$ of \eqref{eq:snr-scaling} (open, dashed). Dotted: chance agreement $0.5$.}
    \label{fig:survival}
\end{figure}

\textbf{Checked attribution margins.} On the $N_c=10$ models used in the certificate-based evaluation, we also check the deterministic radius condition of Theorem~\ref{thm:attr} on the shared aggregates before dispatch. It holds for $100/100$ CIFAR-10 and $59/60$ GNSS client components, exactly the correctly attributed components in those sets. Thus each correct decision has a verified Hamming-separation margin, without invoking the idealized conditional sphere model.

\textbf{Heterogeneity and partial participation.} Attribution is $100\%$ when evaluated only over stations that participate in the non-IID and partial-participation sweeps. Across the full roster, sampling $30\%$ of clients per round gives $98\%$ attribution on GNSS and $100\%$ on CIFAR-10. Strong label skew has a different effect: some clients lack the classes needed for a five-way episode and therefore cannot train a mark. At $\alpha{=}0.5$, roster-wide attribution is $100\%$ on CIFAR-10 and $79\%$ on GNSS. GNSS falls to $16\%$ at $\alpha{=}0.1$. Appendix~\ref{sub:noniid} reports both denominators.

% ======================================================================
\subsection{Tracing, Collusion, and the Two Carriers}
\label{sec:eval_tracing}

\renewcommand{\arraystretch}{0.82}
\begin{table}[t]
\setlength{\tabcolsep}{4pt}
    \centering
    \footnotesize
    \caption{Single-leaker isolation of the dispatched copy (weight-space overlay), 10-seed. \emph{Isolated}: trials whose owner the Tardos accusation names exactly with no innocent accused ($N_c{\times}10$ seeds). \emph{Aggr.\ attr.}: the shared-model argmax attribution from Table~\ref{tab:survival}. $q_{\mathrm{owner}}$: the owner's per-bit flip rate.}
    \label{tab:isolation}
    \begin{tabular}{lccccc}
    \toprule
    \textbf{Dataset} & $N_c$ & $\rho$ & \textbf{Isolated} & $q_{\mathrm{owner}}$ & \textbf{Aggr.\,attr.} \\
    \midrule
    \multirow{2}{*}{GNSS} & 10 & 0.27 & \textcolor{ForestGreen}{$100/100$} & $0.12$ & $0.98$ \\
     & 40 & 1.07 & \textcolor{ForestGreen}{$400/400$} & $0.05$ & $0.87$ \\
    \midrule
    \multirow{2}{*}{CIFAR-10} & 10 & 0.27 & \textcolor{ForestGreen}{$100/100$} & $0.15$ & $1.00$ \\
     & 40 & 1.07 & \textcolor{ForestGreen}{$400/400$} & $0.06$ & $1.00$ \\
    \bottomrule
    \end{tabular}

    \vspace{2pt}
    {\footnotesize Innocent baseline $q{\approx}0.29$--$0.31$ at $N_c{=}10$ and $0.25$--$0.26$ at $N_c{=}40$ (the minimum over the $N_c{-}1$ innocents is an order statistic that falls with $N_c$). The owner sits far below it.}
\end{table}

\textbf{Isolating the leaked copy.} The tracer decodes the dispatched copy's overlay and compares it with the registered secret rows, without the holder's cooperation. Every single-leaker trial identifies exactly the owner: $100/100$ per dataset at $N_c{=}10$ and $400/400$ at $N_c{=}40$ (Table~\ref{tab:isolation}). Isolation therefore persists beyond the identity-layer rank reference $N^{\!*}{=}37.5$. Registry adjudication succeeds in $10/10$ trials. The credential proof accepts the legitimate station and rejects every tested forgery, allowing a third party to check the evidence.

\renewcommand{\arraystretch}{0.82}
\begin{table}[t]
\setlength{\tabcolsep}{4pt}
    \centering
    \footnotesize
    \caption{Collusion traceability under Boneh--Shaw copy averaging, both carriers, $N_c{=}10$, 10-seed pooled (GNSS first). Decisions use score thresholds without a certificate requirement. Traceability: fraction of sampled coalitions with ${\geq}1$ colluder accused and no innocent accused. \emph{Framed}: fraction of coalitions in which any innocent is accused, at every $k$.}
    \label{tab:collusion}
    \begin{tabular}{llccccc}
    \toprule
    \textbf{Carrier} & \textbf{Dataset} & $k{=}2$ & $k{=}3$ & $k{=}5$ & $k{=}8$ & \textbf{Framed} \\
    \midrule
    \multirow{2}{*}{Weight ($c{=}2$)$^{\S}$} & GNSS & $1.00$ & $0.95$ & $0.40$ & $0.15$ & \textcolor{ForestGreen}{$0.000$} \\
     & CIFAR-10 & $1.00$ & $0.75$ & $0.35$ & $0.00$ & \textcolor{ForestGreen}{$0.000$} \\
    \midrule
    \multirow{2}{*}{Feature ($c{=}3$)} & GNSS & $1.00$ & $1.00$ & $1.00$ & $0.76$ & \textcolor{ForestGreen}{$0.000$} \\
     & CIFAR-10 & $\mathbf{1.00}$ & $\mathbf{1.00}$ & $\mathbf{1.00}$ & $\mathbf{1.00}$ & \textcolor{ForestGreen}{$0.000$} \\
    \bottomrule
    \end{tabular}

    \vspace{2pt}
    {\footnotesize $^{\S}$~Weight-carrier traceabilities use the implemented central-limit threshold. At the larger candidate threshold $Z{=}97.1$, the mid-coalition cells fall to $0.65$/$0.50$ at $k{=}3$ and $0.10$/$0.05$ at $k{=}5$ (GNSS/CIFAR-10). The $c{=}2$ design point ($1.00$) and the zero-framing column are invariant to the threshold.}
\end{table}

\begin{figure}[t]
    \centering
    \includegraphics[width=\columnwidth]{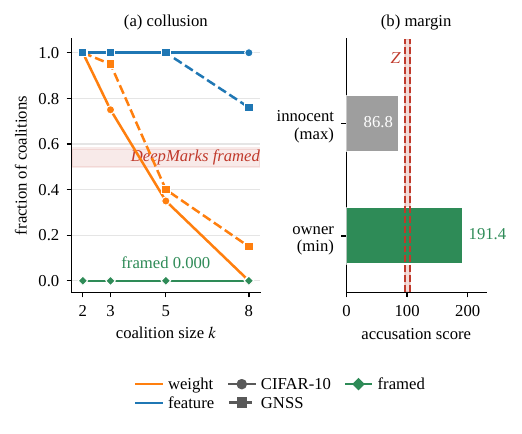}
    \caption{Collusion tracing and the accusation margin. (a)~Fraction of sampled coalitions traced, and the fraction in which any innocent is framed, against coalition size $k$ for both carriers and datasets (Table~\ref{tab:collusion}). Color keys the carrier, marker and line style the dataset. Shaded band: the DeepMarks-BIBD \emph{framing} rate, whose traceability is not plotted. (b)~Largest measured innocent score and smallest owner score over the realized decodes. Shaded band: the accusation threshold $Z{=}\sqrt{2\nt\ln(N/\varepsilon_1)}$ over the deployed federation sizes.}
    \label{fig:tracing}
\end{figure}

\textbf{Collusion tracing and false accusations.} The weight overlay traces every tested two-client coalition, at its design target $c{=}2$. Success declines for larger coalitions (Table~\ref{tab:collusion}, Fig.~\ref{fig:tracing}). The feature carrier is designed for $c{=}3$ and traces every sampled coalition through $k{=}5$ on GNSS and $k{=}8$ on CIFAR-10. These outcomes are empirical. The finite channel guarantees are given in Table~\ref{tab:finite-guarantees}. Neither carrier accuses an innocent station at any tested coalition size, including the $N_c{=}40$ experiments.

\textbf{Certificate-based tracing.} We evaluate $N_c=10$ weight-carrier models over ten CIFAR-10 seeds and six GNSS seeds, testing all ten single copies and all $\binom{10}{2}=45$ pair averages per seed. The certificate-gated rule allocates $10^{-3}$ per investigation across accusation and tamper tails. All $880$ positive-tail certificates pass. The adjudication rule isolates $160/160$ single copies and traces $712/720$ pair mixtures without naming an innocent (Table~\ref{tab:certified-audit}). The eight missed mixtures are CIFAR-10 pairs. This evaluation enumerates every pair within each seed. The coalition sweep in Table~\ref{tab:collusion} samples coalitions and uses threshold-only decisions.

\begin{table}[t]
\centering\footnotesize
\caption{Certificate-based weight-carrier tracing. Each entry counts decisions naming a colluder and no innocent. Positive-tail certificates pass for every tested decode. No innocent is named on either tail.}
\label{tab:certified-audit}
\begin{tabular}{lccc}
\toprule
Dataset & Seeds & Single copies & All pair averages\\
\midrule
GNSS & 6 & 60/60 & 270/270\\
CIFAR-10 & 10 & 100/100 & 442/450\\
\bottomrule
\end{tabular}
\end{table}

Each certificate establishes its numerical bound under A5(a), without requiring a binary-symmetric residual channel. Decoded arrays and interval enclosures accompany the results so that each bound can be checked independently. It does not empirically establish the independence premise. Unanimous-position disagreement averages $0.112$ on GNSS pairs and $0.140$ on CIFAR-10 pairs. These are measured channel diagnostics, not estimates validating the modeled $q$ values in Table~\ref{tab:finite-guarantees}.

\renewcommand{\arraystretch}{0.82}
\begin{table}[t]
\setlength{\tabcolsep}{4pt}
    \centering
    \footnotesize
    \caption{Two-carrier comparison. Where two figures appear they are GNSS\,/\,CIFAR-10. Tracing lengths: $512$ weight and $2048$ feature. Max $k$: largest coalition traced at $100\%$. Copy cost: accuracy drop on the dispatched copy. Global cost: accuracy change on the deployed model.}
    \label{tab:twocarrier}
    \begin{tabular}{lcc}
    \toprule
    \textbf{Property} & \textbf{Weight-space} & \textbf{Feature-space} \\
    \midrule
    Collusion design target $c$ & $2$ & $3$ \\
    Max $k$ traced at $100\%$    & $2$\,/\,$2$ & $5$\,/\,$8$ \\
    KD survival (feature-match)  & \textcolor{red}{$0/10$} & \textbf{\textcolor{ForestGreen}{$10/10$}} \\
    Copy-accuracy cost (pp)      & ${\approx}4.5$ & $4.8$\,/\,$6.1$ \\
    Global-accuracy cost (pp)    & $0$ & $0$ \\
    \bottomrule
    \end{tabular}
\end{table}

\textbf{Carrier tradeoff.} Both carriers encode the secret tracing row and use the same registry and credential-authentication procedure (Table~\ref{tab:twocarrier}). The weight carrier targets two colluders and has a dispatched-copy accuracy cost of approximately $4.5$ points in the reported $300$-episode, $\lambda_t=6$ configuration. Distillation erases this carrier. The feature carrier targets three colluders and survives the tested feature-matching distillation runs, with a sufficient margin certificate given by Theorem~\ref{thm:score-stability}. Its copy-accuracy cost is $4.8$ points on GNSS and $6.1$ on CIFAR-10. Function-only distillation erases this carrier too. Both tracing marks are added at dispatch, so neither changes the global model.

% ======================================================================
\subsection{Comparison, Utility, and Robustness}
\label{sec:eval_robustness}

\renewcommand{\arraystretch}{0.82}
\begin{table}[t]
\setlength{\tabcolsep}{3pt}
    \centering
    \footnotesize
    \caption{Feature-space metrics, unwatermarked (Base) vs.\ watermarked (WM) backbone (10-seed mean$\pm$std, GNSS first). $1$-NN LOO: leave-one-out nearest-neighbor accuracy.}
    \label{tab:featurespace}
    \resizebox{\columnwidth}{!}{%
    \begin{tabular}{lcccc}
    \toprule
    & \multicolumn{2}{c}{\textbf{GNSS}} & \multicolumn{2}{c}{\textbf{CIFAR-10}} \\
    \cmidrule(lr){2-3}\cmidrule(lr){4-5}
    \textbf{Metric} & \textbf{Base} & \textbf{WM} & \textbf{Base} & \textbf{WM} \\
    \midrule
    Silhouette          & $0.605_{\pm.027}$ & $0.556_{\pm.040}$ & $0.340_{\pm.002}$ & $0.318_{\pm.004}$ \\
    $1$-NN LOO acc.     & $0.905_{\pm.010}$ & $0.911_{\pm.013}$ & $0.794_{\pm.002}$ & $0.797_{\pm.003}$ \\
    Intra-class dist.   & $3.36_{\pm.25}$   & $5.01_{\pm.39}$   & $7.31_{\pm.11}$   & $10.16_{\pm.18}$  \\
    Inter-class dist.   & $12.60_{\pm.53}$  & $16.35_{\pm.50}$  & $14.53_{\pm.17}$  & $18.50_{\pm.31}$  \\
    \bottomrule
    \end{tabular}}
\end{table}

\begin{table*}[t]
\setlength{\tabcolsep}{3pt}
\renewcommand{\arraystretch}{0.86}
    \centering
    \footnotesize
    \caption{Federated-watermarking baselines ($N_c{=}10$, 10-seed mean$\pm$std). Capability columns are defined in the text. \textcolor{ForestGreen}{$\checkmark$}/\textcolor{red}{$\times$} denote presence/absence, \emph{partial} a qualified guarantee, and a parenthesis names the party the guarantee rests on. Cost: one-time verification payload. Acc: few-shot accuracy. Mark: each method's native mark metric, DeepMarks-BIBD being scored on its native Boolean tracing only, hence the empty cells.}
    \label{tab:baselines}
    \begin{tabular}{lcccccccccc}
    \toprule
    & & & & & & & \multicolumn{2}{c}{\textbf{GNSS}} & \multicolumn{2}{c}{\textbf{CIFAR-10}} \\
    \cmidrule(lr){8-9}\cmidrule(lr){10-11}
    \textbf{Method} & \textbf{Attribution} & \textbf{Credential} & \textbf{ZK} & \textbf{KD} & \textbf{Coll.} & \textbf{Cost} & \textbf{Acc} & \textbf{Mark} & \textbf{Acc} & \textbf{Mark} \\
    \midrule
    \textbf{Ours} & \textcolor{ForestGreen}{$\checkmark$} & \textcolor{ForestGreen}{$\checkmark$} & \textcolor{ForestGreen}{$\checkmark$} & \textcolor{ForestGreen}{$\checkmark^{\ast}$} & \textcolor{ForestGreen}{$\checkmark$} & 0.6\,MB & $0.934_{\pm.010}$ & $\mathbf{0.98}_{\pm.04}$ & $0.849_{\pm.003}$ & $\mathbf{1.00}_{\pm.00}$ \\
    FedZKP & \textcolor{red}{$\times$}\,(group) & \textcolor{ForestGreen}{$\checkmark$} & \textcolor{ForestGreen}{$\checkmark$} & \textcolor{red}{$\times$} & \textcolor{red}{$\times$} & interactive & $0.939_{\pm.004}$ & $0.999_{\pm.001}$ & $0.847_{\pm.002}$ & $0.986_{\pm.003}$ \\
    FedTracker & \textcolor{ForestGreen}{$\checkmark$} & \textcolor{red}{$\times$}\,(server) & \textcolor{red}{$\times$} & \textcolor{red}{$\times$} & \textcolor{red}{$\times$} & 2.3\,MB & $0.926_{\pm.009}$ & $1.00_{\pm.00}$ & $0.838_{\pm.005}$ & $1.00_{\pm.00}$ \\
    DUW & \textcolor{ForestGreen}{$\checkmark$} & \textcolor{red}{$\times$}\,(server) & \textcolor{red}{$\times$} & partial & \textcolor{red}{$\times$} & 1.5\,MB & $0.922_{\pm.019}$ & $1.00_{\pm.00}$ & $0.825_{\pm.018}$ & $1.00_{\pm.00}$ \\
    FedIPR & partial & partial & \textcolor{red}{$\times$} & \textcolor{red}{$\times$} & \textcolor{red}{$\times$} & 2.3\,MB & $0.925_{\pm.016}$ & $0.997_{\pm.002}$ & $0.849_{\pm.003}$ & $0.980_{\pm.003}$ \\
    WAFFLE & \textcolor{red}{$\times$} & \textcolor{red}{$\times$} & \textcolor{red}{$\times$} & partial & \textcolor{red}{$\times$} & 0.3\,MB & $0.933_{\pm.009}$ & $0.218_{\pm.126}$ & $0.846_{\pm.003}$ & $0.289_{\pm.035}$ \\
    \midrule
    DeepMarks-BIBD & \textcolor{ForestGreen}{$\checkmark$} & \textcolor{red}{$\times$}\,(server) & \textcolor{red}{$\times$} & \textcolor{red}{$\times$} & partial & -- & -- & -- & -- & -- \\
    \bottomrule
    \end{tabular}

    \vspace{2pt}
    {\footnotesize $^{\ast}$~Feature carrier survives feature-matching KD ($10/10$ per dataset), but not function-only KD. All tested weight-space marks fail.}
\end{table*}

Table~\ref{tab:baselines} compares the five baselines using their native mark metrics. The capability columns denote cooperation-free recipient identification (\emph{Attribution}), client-secret proof of knowledge (\emph{Credential}), verification without disclosure of reusable secrets (\emph{ZK}), distillation survival (\emph{KD}), and coded tracing with a conditional false-accusation bound (\emph{Coll.}). ZK-Trace combines these capabilities subject to the carrier, channel, and per-decode qualifications above.

FedZKP and FedIPR verify ownership of a shared model rather than identify its recipient. WAFFLE's group-mark metric is weak on both datasets ($0.22$ GNSS, $0.29$ CIFAR-10). FedTracker and DUW achieve per-client traceability $1.0$, but their fingerprints are server-assigned and lack binding to a client secret or zero-knowledge verification. Their distillation outcomes depend on method and dataset (App.~\ref{sub:kd-leveling}). DeepMarks-BIBD~\cite{chen_deepmarks} addresses collusion under exact extraction, an assumption disrupted by FedAvg bit errors. It frames innocents in $50\%$ of plain coalitions and $58\%$ within its design resilience. ZK-Trace frames none in the same sweep (Fig.~\ref{fig:tracing}).

\textbf{Global-model utility.} Dispatch fingerprinting does not modify the global model. For the identity watermark, paired ten-seed watermark-off/on comparisons establish equivalence at the stated margins on both datasets: $\Delta=-0.04$\,pp on GNSS (margin $1.0$\,pp, TOST $p=0.0052$) and $+0.24$\,pp on CIFAR-10 (margin $0.5$\,pp, $p=0.0489$). Appendix~\ref{sub:utility-ext} gives confidence intervals and the matched-artifact protocol.

\textbf{Feature-space geometry.} Table~\ref{tab:featurespace} evaluates the representation learned with the identity watermark. Leave-one-out $1$-NN accuracy, a diagnostic of local class separation, shows no significant change on either dataset. Silhouette scores decrease by $0.048$ on GNSS ($p{=}0.009$) and $0.023$ on CIFAR-10 ($p{<}10^{-4}$). Both intra- and inter-class distances increase. The observed pattern indicates altered cluster compactness alongside similar nearest-neighbor classification. It does not imply that watermarking leaves feature geometry unchanged.

\renewcommand{\arraystretch}{0.80}
\begin{table}[t]
\setlength{\tabcolsep}{3pt}
    \centering
    \caption{Robustness of \emph{attribution} for the weight-space carrier under the parameter-space and training attacks ($N_c{=}10$, 100 test episodes, 10-seed mean$\pm$std, GNSS first). Self: self bit-accuracy (${\approx}0.5$ = codeword erased). Acc: few-shot accuracy (\%). Attr.: attribution accuracy. \textcolor{ForestGreen}{Green}/\textcolor{red}{red}: attribution preserved / broken.}
    \label{tab:attacks}
    \resizebox{\columnwidth}{!}{%
    \begin{tabular}{lcccccc}
    \toprule
    & \multicolumn{3}{c}{\textbf{GNSS}} & \multicolumn{3}{c}{\textbf{CIFAR-10}} \\
    \cmidrule(lr){2-4} \cmidrule(lr){5-7}
    \textbf{Attack} & \textbf{Self} & \textbf{Acc} & \textbf{Attr.} & \textbf{Self} & \textbf{Acc} & \textbf{Attr.} \\
    \midrule
    None (clean)               & $0.96_{\pm.02}$ & $93.4_{\pm1.0}$ & \textcolor{ForestGreen}{$0.98_{\pm.04}$} & $0.94_{\pm.01}$ & $84.9_{\pm.3}$ & \textcolor{ForestGreen}{$1.00_{\pm.00}$} \\
    Structured prune 50\%      & $0.95_{\pm.02}$ & $31.2_{\pm2.8}$ & \textcolor{ForestGreen}{$0.98_{\pm.04}$} & $0.94_{\pm.01}$ & $52.0_{\pm3.6}$ & \textcolor{ForestGreen}{$1.00_{\pm.00}$} \\
    Noise $\sigma{=}2.0$         & $0.78_{\pm.01}$ & $26.7_{\pm1.9}$ & \textcolor{ForestGreen}{$0.98_{\pm.04}$} & $0.76_{\pm.01}$ & $26.0_{\pm.8}$ & \textcolor{ForestGreen}{$1.00_{\pm.00}$} \\
    Quantize 2-bit             & $0.93_{\pm.02}$ & $29.4_{\pm5.0}$ & \textcolor{ForestGreen}{$0.98_{\pm.04}$} & $0.92_{\pm.01}$ & $29.4_{\pm5.4}$ & \textcolor{ForestGreen}{$1.00_{\pm.00}$} \\
    PGD ($\epsilon{=}0.1$)       & $0.95_{\pm.02}$ & $82.4_{\pm9.1}$ & \textcolor{ForestGreen}{$0.98_{\pm.04}$} & $0.94_{\pm.01}$ & $80.4_{\pm.8}$ & \textcolor{ForestGreen}{$1.00_{\pm.00}$} \\
    Reset layer3+4             & $0.59_{\pm.02}$ & $41.9_{\pm4.6}$ & $0.63_{\pm.23}$ & $0.56_{\pm.02}$ & $26.3_{\pm2.2}$ & $0.48_{\pm.17}$ \\
    Fine-tune 100\,ep          & $0.96_{\pm.02}$ & $82.2_{\pm8.4}$ & \textcolor{ForestGreen}{$0.98_{\pm.04}$} & $0.94_{\pm.01}$ & $79.1_{\pm.7}$ & \textcolor{ForestGreen}{$1.00_{\pm.00}$} \\
    KD 80\,ep                  & $0.50_{\pm.01}$ & $80.4_{\pm10.8}$ & \textcolor{red}{$0.15_{\pm.11}$} & $0.51_{\pm.01}$ & $80.6_{\pm1.9}$ & \textcolor{red}{$0.15_{\pm.08}$} \\
    \bottomrule
    \end{tabular}}
\end{table}

\textbf{Removal attacks.} Pruning, noise, quantization, PGD, and fine-tuning preserve weight-carrier attribution at its clean level in Table~\ref{tab:attacks}. Layer reset reduces attribution, but also severely degrades task accuracy. Distillation is the tested attack that erases the weight mark while retaining approximately $80\%$ task accuracy (Fig.~\ref{fig:tradeoff}). All tested weight-space baselines also lose their marks under the same $80$-epoch distillation.

\textbf{Feature-carrier distillation.} Feature-matching distillation preserves the feature mark in $20/20$ runs, ten per dataset. Cross-architecture transfer from ResNet-18 to ResNet-34 preserves it in $19/20$ runs. The single loss is GNSS seed $271$, whose flip rate rises to $0.4976$. Across seeds, the architecture change produces larger and more variable flip-rate increases on GNSS than on CIFAR-10. Sparse four-channel inputs or subtle class differences may contribute, but these experiments do not isolate either cause. The baseline scale $G$ governs the weight-carrier bound, not this feature-carrier failure (App.~\ref{sub:kd-leveling}). Function-only distillation erases the feature mark in all twenty runs. Survival therefore depends on the feature-matching condition in Assumption~\ref{ass:kd}.

% Declared BEFORE the paragraph that references it, unlike fig:survival above.
% Sec. VII-C ends close to a page boundary, so a [!t] figure declared after this
% paragraph has no top slot left on the current page and LaTeX defers it a page,
% stranding it past Sec. VIII and costing a whole page of the 13-page budget.
% Moving the declaration earlier gives it a top slot on the same page as
% tab:attacks. Source position matters here -- [!t]/[!tb] alone do not fix it.
\begin{figure}[!t]
    \centering
    \includegraphics[width=\columnwidth]{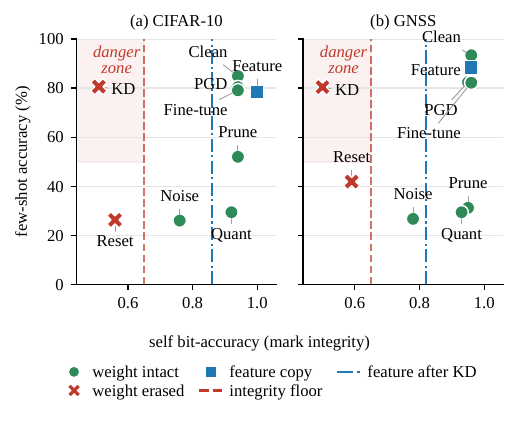}
    \caption{Accuracy--integrity tradeoff under the weight-carrier attacks of Table~\ref{tab:attacks}: (a)~CIFAR-10, (b)~GNSS. $x$-axis: self bit-accuracy (${\approx}0.5$ = codeword erased). $y$-axis: few-shot accuracy. Shaded: high accuracy with the codeword at chance. Square: clean feature-carrier copy. Dash-dot: its bit-accuracy after the same $80$-epoch distillation.}
    \label{fig:tradeoff}
\end{figure}

\textbf{Own-row erasure and hybrid escape.} The threshold-only tests detect complementary single-strategy erasures (Prop.~\ref{prop:adjudicate}, Lem.~\ref{lem:tamper-feature}). A hybrid attack defeats that complementarity by negating the weight overlay while independently re-randomizing the feature mark. At matched compute, it escapes the threshold-only tests on both carriers in $5/20$ runs.

The hybrid attack moves the weight score along a line with two decision thresholds. Scores above $Z$ trigger accusation, while sufficiently negative scores trigger the reflected tamper test. Moderate negation moves the score into the interval between these regions, while stronger negation can push it through the interval and into tamper detection. At the matched $300$-episode split, $15/20$ copies remain above the positive threshold. All twenty copies escape at twice that budget, but all twenty trigger tamper detection at four times the budget. Under the stated custody model, the attacker cannot directly monitor the tracer-secret score, so the tested attack selects compute without observing its stopping point. These results characterize an unresolved evasion route. They do not establish robustness to an adaptively chosen budget. No innocent is flagged at any tested budget (App.~\ref{ssub:adjudicate}).

\textbf{Deployment stressors.} Coordinate-wise median aggregation preserves attribution ($1.00$ CIFAR-10, $0.98$ GNSS), although task accuracy declines. Under GNSS receiver covariate shift, tracing succeeds in all $100$ trials and aggregate attribution remains $0.98$, while exact single-leaker isolation degrades toward chance.

\textbf{Per-class copy cost.} On the dispatched copies, mean per-class recall falls by $4.4$ points on CIFAR-10 and $10.9$ on GNSS, with larger losses on the harder interference classes. These values use the original label space rather than five-way episodes and therefore differ from the episode-accuracy costs in Table~\ref{tab:twocarrier} (App.~\ref{sub:stressors}).

%% file: 08_conclusion.tex
\section{CONCLUSION}
\label{label_conclusion}

ZK-Trace combines recipient tracing, client-secret credentials, and independently checkable false-accusation certificates for federated GNSS models. The analysis provides finite completeness bounds for whole-codeword coalition strategies under a specified residual channel, together with deterministic score-stability bounds for correlated feature errors. At the two tracing lengths, verified finite examples give missed-coalition bounds below $0.0096$ and $0.00201$ with false-accusation budget $10^{-3}$. These modeled-channel results complement the certificate-based evaluation: all $880$ positive-tail certificates pass, all $160$ single copies are isolated, and $712/720$ pair mixtures are traced without naming an innocent.

Carrier choice controls the robustness and utility tradeoff. Feature-space tracing survives the tested feature-matching distillation runs, while function-only distillation and hybrid erasure remain evasion routes. Dispatch fingerprints leave the shared model unchanged. Matched controls establish global-model utility equivalence on both datasets at the stated margins. The remaining deployment work is to validate the conditional independence and probe-margin premises across independently sited GNSS stations and to address the demonstrated evasion routes.

\section*{Data and Code Availability}
Code, per-seed artifacts, and configuration files covering both carriers, the Tardos code, the registry and verifier, the attack suite, and the table-reproduction scripts will be released upon publication. The GNSS tensors are re-rendered from the interference recordings released by their originators~\cite{heublein_feigl_crpa}.

%% file: 09_appendix.tex
\clearpage
\appendices
\section{Tracing Construction and Analysis}
\label{app:method-detail}

The construction separates credential authentication, identity decoding, and recipient tracing. The bounds below specify their different probability spaces: identity code generation, an innocent row conditional on the recovered artifact, and a coalition channel with hidden biases. Experiments use fixed reproducible code seeds. Their success counts are distinct from these probability bounds.

\subsection{Identity Codebook and Tardos Overlay}
\label{label_method_collusion}

The identity codewords $w_i$ occupy separate column blocks $E_i$. Every block is a set of dense directions in the same $\omega$-dimensional scale vector. They do not occupy disjoint physical coordinates. An additional shared block $E_T$ carries the dispatched recipient's tracing row. The identity layer detects agreement with public credentials. It does not establish who trained the model.

\begin{definition}[Tardos overlay on a shared block]
\label{def:overlay}
Fix a design target $c\ge2$, tracing length $\nt$, and $\delta_c=1/(300c)$. Draw biases independently with density
\[
f_{\delta_c}(p)=\frac{1}{(\pi-4\arcsin\sqrt{\delta_c})\sqrt{p(1-p)}}
\]
on $[\delta_c,1-\delta_c]$. Conditional on the biases, draw all entries $X_{i,b}\sim\mathrm{Bernoulli}(p_b)$ independently. At dispatch, embed row $X_i$ in recipient $i$'s copy using the shared unit columns $E_{T,b}$. Decode $y_b=\ind\{\inner{\gamma_{\mathrm{pir}}}{E_{T,b}}>0\}$ and score
\begin{equation}
\label{eq:U-def}
S_i=\sum_{b=1}^{\nt}U(X_{i,b},y_b,p_b),\quad
U(x,y,p)=\frac{(2y-1)(x-p)}{\sqrt{p(1-p)}}.
\end{equation}
\end{definition}

The score is the symmetric Tardos score~\cite{skoric2008symmetric,tardos2008optimal}. It satisfies
\begin{equation}
\label{eq:antisym}
\begin{split}U(x,1-y,p)&=-U(x,y,p),\\ |U|&\le B:=\sqrt{(1-\delta_c)/\delta_c}.\end{split}
\end{equation}
The carrier uses $\nt$ shared columns rather than $N\nt$ client-specific columns. The full code matrix and biases are held by the tracer. Clients can estimate their own rows from their copies and may receive their rows in the insider-attack model. No secrecy from one's own recipient is needed for the innocent-row analysis.

\subsection{Dispatch Fingerprinting}
\label{label_method_dispatch}

Federated training embeds the public identity layer. The tracer subsequently fine-tunes a separate copy for each recipient on a server-held proxy pool, with
\begin{equation}
\label{eq_dispatch_hinge}
\mathcal L_i=\mathcal L_{\mathrm{task}}+\frac{\lambda_t}{\nt}
\sum_b\max\{0,\mu-(2X_{i,b}-1)\inner{\gamma}{E_{T,b}}\}.
\end{equation}
The reported weight configuration uses roughly $300$ episodes and $\lambda_t=6$. The task term preserves classifier utility during embedding. The copy is then evaluated separately from the shared model. The tracing operation does not change that shared model. Table~\ref{tab:twocarrier} reports copy costs for the evaluated configurations, not a universal cost of embedding.

The codebook and tracing overlay have different purposes. The global model contains multiple identity marks, so those marks cannot isolate a recipient. The dispatched copy contains the recipient's tracing row, allowing offline comparison with the enrollment registry. Neither layer detects poisoned updates or prevents voluntary credential sharing. These threats require separate controls, such as robust aggregation and revocation~\cite{yin2018byzantine,naor2001revocation}.

\subsection{Finite Tracing Bounds and the Residual Channel}
\label{label_method_flipchannel}

\begin{lemma}[Innocent-score moments]
\label{lem:innoc}
Under A5(a), for an innocent $i$ and any realized $(y,p)$,
\begin{equation}
\label{eq:innoc-moments}
\E[S_i\mid y,p]=0,\qquad \operatorname{Var}(S_i\mid y,p)=\nt.
\end{equation}
The individual increments are independent and bounded by $B$. Their conditional distribution and upper tails may depend on $y$.
\end{lemma}

\begin{lemma}[Flip-scaled moments]
\label{lem:coalmean}
Let $T_b=\sum_{i\in\Coll}U(X_{i,b},y_b^*,p_b)$ and let $F_b\sim\mathrm{Bernoulli}(q)$ be independent of $T_b$. Then $T_b'=(1-2F_b)T_b$ obeys
\begin{equation}
\label{eq:coal-mean}
\begin{split}
\E[T_b']&=(1-2q)\E[T_b],\\
\E[(T_b')^2]&=\E[T_b^2],\\
\operatorname{Var}(T_b')&=\operatorname{Var}(T_b)+4q(1-q)(\E[T_b])^2.
\end{split}
\end{equation}
\end{lemma}

Here $q$ denotes disagreement with the intended coalition output $y^*$, not with each colluder's row. On positions where colluders disagree, a noiseless marking-consistent output already differs from some rows. Consequently, the reported per-row disagreement rates do not identify a residual BSC parameter.

\begin{lemma}[A deterministic weight-carrier perturbation bound]
\label{lem:flip}
For a fixed noiseless mixture $\gamma^*=\sum_{i\in\Coll}\lambda_i\gamma^{(i)}$ and perturbation $\eta$, let $a_b=\inner{\gamma^*}{E_{T,b}}$. Any position with $|a_b|>\|\eta\|_2$ keeps its sign after adding $\eta$.
\end{lemma}

The unit-column Cauchy--Schwarz bound proves this statement without a Gaussian or independence assumption. Equal and opposite colluder margins may cancel, so local margins do not guarantee a positive pooled margin on every bit. Because averaging neural-network weights does not imply averaging feature responses, the feature-carrier analysis uses probe margins directly.

\begin{remark}[A certificate is an evaluated inequality]
\label{rem:deployed-vs-provable}
Theorem~\ref{thm:tardos} gives a conditional bound for each $(y,p)$ at any positive $\alpha$. The checker selects an $\alpha$ numerically and encloses the resulting exponent using interval arithmetic. A passing certificate requires its lower endpoint to exceed the upper endpoint of $\ln(N/\varepsilon)$ for the allocated tail budget $\varepsilon$. This avoids assuming that an optimizer has found the exact supremum. An uncertified threshold exceedance remains an investigative lead. A valid certificate does not prove that the innocent-row independence premise holds in a deployment.
\end{remark}

The finite completeness bound integrates over the hidden biases before maximizing a coalition's allowed output. This order is essential: a pirate knows its rows but not the biases.

\begin{theorem}[Finite completeness against row-dependent marking strategies]
\label{thm:finite-completeness}
Fix a coalition of $s\le c$ rows from Definition~\ref{def:overlay}. It may choose its full intended word as any randomized function of those rows and side information independent of the biases conditional on the rows. At unanimous positions it must output the common bit. Apply independent residual BSC flips with common rate $q$. For $k\in\{0,\ldots,s\}$ set
\[
\begin{gathered}
a_k(p)=\frac{k-sp}{\sqrt{p(1-p)}},\qquad \mathcal Y_0=\{0\},\quad \mathcal Y_s=\{1\},\\
\mathcal Y_k=\{0,1\}\quad(0<k<s).
\end{gathered}
\]
For $t>0$ define
\begin{equation}
\label{eq:finite-coal-mgf}
\begin{split}
J_s(t,q)=\sum_{k=0}^{s}\binom{s}{k}\max_{v\in\mathcal Y_k}
\int_{\delta_c}^{1-\delta_c}&f_{\delta_c}(p)p^k(1-p)^{s-k}\\[-2pt]
{}\cdot\big[(1-q)e^{-t(2v-1)a_k(p)}&+q e^{t(2v-1)a_k(p)}\big]\,dp.
\end{split}
\end{equation}
At a fixed threshold $z$, the probability of missing every colluder satisfies
\begin{equation}
\label{eq:finite-completeness}
\Prob[\max_{i\in\Coll}S_i\le z]\le
\min\{1,e^{tsz}J_s(t,q)^{\nt}\}.
\end{equation}
This holds even when the intended symbols are dependent across positions.
\end{theorem}

The bound can be checked for every $s\le c$ and optimized over $t$. It is finite at the implemented cutoff and specifies its own completeness level. For a compatible fixed threshold, define
\begin{equation}
\label{eq:finite-innoc-mgf}
I(a)=\int_{\delta_c}^{1-\delta_c}f_{\delta_c}(p)
\max_{v\in\{0,1\}}M(a;v,p)\,dp.
\end{equation}
Under A5(a), $N e^{-az}I(a)^{\nt}\le\varepsilon_1$ suffices for a-priori soundness. This bound averages over independently generated biases. Theorem~\ref{thm:tardos}, by comparison, conditions on the realized biases. Neither bound requires the innocent score distribution to be invariant under changes of output.

\begin{table}[t]
\centering\footnotesize
\caption{Finite guarantees for the specified residual channels, $N=10$, $\varepsilon_1=10^{-3}$, $\delta_c=1/(300c)$. The completeness bound covers every coalition size $1\le s\le c$. Displayed thresholds are rounded. The calculation uses the full-precision values in the accompanying artifact.}
\label{tab:finite-guarantees}
\begin{tabular}{ccccc}
\toprule
$c$ & $\nt$ & $q$ & $z$ & Upper bound on $\varepsilon_2$\\
\midrule
2 & 512 & 0.04 & 111.0004 & 0.0096\\
3 & 2048 & 0.15 & 214.9809 & 0.00201\\
\bottomrule
\end{tabular}
\end{table}

\textbf{Verified finite examples.} Table~\ref{tab:finite-guarantees} evaluates both inequalities using 4096 interval panels in the arcsine angle coordinate and 30-decimal-digit interval arithmetic. Numerical optimization selects candidate values of $(a,t)$. Interval integration then encloses the integrals from above to establish the displayed bounds. All $s=1,\ldots,c$ are checked. The residual rates are stipulated model parameters, not fitted neural extraction rates. The script and full-precision enclosures are supplied with the reproducibility artifacts.

\textbf{Relation to the design equation.} The common planning rule
\begin{equation}
\label{eq:design}
\nt^{\mathrm{design}}=\left\lceil d c^2(1-2q)^{-2}\ln(N/\varepsilon_1)\right\rceil,
\qquad d=\pi^2/2,
\end{equation}
is an asymptotic estimate~\cite{skoric2008symmetric}. Laarhoven--de~Weger~\cite{laarhoven2014optimal} obtain a finite constant $23.79$ jointly with threshold coefficient $8.06$ and cutoff coefficient $28.31$. Those constants jointly specify a different code configuration. Equations~\eqref{eq:finite-coal-mgf}--\eqref{eq:finite-innoc-mgf} give soundness and completeness bounds for cutoff coefficient $300$ and the thresholds used here. Alternative codes and decoders offer other tradeoffs~\cite{nuida2009improvement,meerwald2012joint,oosterwijk2013optimal,hollmann1998ipp}.

\subsection{Registry, Authentication, and Disputes}
\label{label_method_registry}

The registry binds the identity and tracing assignments before a dispute:
\begin{equation}
\label{eq_enroll}
C_i=\Com(w_i\|X_i\|A_i\|y_i;\rho_i).
\end{equation}
The tracer creates the enrollment tuple, and the judge holds its opening. Tracing needs no response from the suspect. A judge checks the opening, registered decoder setup, and score evidence. A commitment prevents substitution of a different row, but a tracer that knows the original row can embed it again. As in the evaluated protocol, provenance therefore assumes an honest tracer. A signed delivery receipt and authenticated setup would be needed to prove issuance independently. A client nonce alone does not prove delivery.

The credential proof authenticates a claimant in a dispute. Both an innocent recipient and a genuine leaker can possess a valid witness, so witness knowledge alone is not exculpatory evidence. Exculpation requires rejection of an invalid accusation, such as a mismatched enrollment opening or a failed tracing certificate. A credential-only authentication proof can omit model presence. Algorithm~\ref{alg_verify} checks both credential knowledge and presence. The implementation's registry helper checks the opening. A deployment's judge must independently reconstruct the score and setup from the evidence package.

\begin{table}[t]
\centering\footnotesize
\caption{Custody and scope of the guarantees.}
\label{tab:secrecy}
\begin{tabular}{p{.25\columnwidth}p{.25\columnwidth}p{.32\columnwidth}}
\toprule
Artifact & Custody & Consequence of disclosure\\\midrule
xLPN witness & client & credential impersonation\\
Recipient row & tracer; recipient can estimate it & own-row removal is possible\\
Full code and biases & tracer & innocent-row independence may fail\\
Projection directions & public & targeted carrier edits are possible\\
Registry opening & judge/tracer & integrity still depends on binding\\
\bottomrule
\end{tabular}
\end{table}

\subsection{Certified Decisions and Residual Erasure}
\label{ssub:adjudicate}

\begin{corollary}[Budgeted two-tail decisions]
\label{cor:secret-row}
For $K$ declared carriers and per-investigation budget $\varepsilon$, allocate $\varepsilon/(2K)$ to each tail. For each carrier, certify the positive tail using $y$ and the negative tail using $1-y$. Under A5(a), a union bound gives probability at most $\varepsilon$ that any certified decision names an innocent, across all $K$ carriers and both tails. Independence between carriers is unnecessary.
\end{corollary}

\begin{proposition}[Certificate-gated adjudication]
\label{prop:adjudicate}
Use a candidate threshold $z=\sqrt{2\nt\ln(2KN/\varepsilon)}$ on each tail. Return \textsc{certified-attribute} if a positive score exceeds $z$ and its tail certificate passes. Otherwise, return \textsc{certified-tamper} for a negative score below $-z$ with a passing reflected certificate. A threshold exceedance without its certificate is an \textsc{uncertified-lead}. If neither threshold is exceeded, return \textsc{no-certified-evidence}. Positive decisions take precedence, with the largest positive score selected. Negative decisions select the smallest score. The probability bound is that of Corollary~\ref{cor:secret-row}.
\end{proposition}

The lead carries no calibrated accusation guarantee. No-certified-evidence does not establish that a copy was never leaked. Repeated investigations need a separately allocated total budget. Adaptive feedback also requires maintaining the conditional independence premise. The attack sweeps measure threshold exceedances. The certificate-based evaluation additionally requires the corresponding tail bound to pass.

\begin{lemma}[Perfect negation reverses the score]
\label{lem:tamper-feature}
For $Y=1-X_i$, $S_i=-\sum_b |X_{i,b}-p_b|/\sqrt{p_b(1-p_b)}$. Its expectation over row $i$, conditional on $p$, is $-2\sum_b\sqrt{p_b(1-p_b)}$. A negative verdict still requires its threshold and certificate.
\end{lemma}

The experiments show different erasure responses on the two carriers. Negating the weight mark can leave a positive score below its threshold, whereas negating feature marks produces large negative scores. Re-randomizing feature signs can bring their scores near zero. These are measured attack outcomes, not a guarantee that every negative-training objective reaches perfect inversion.

\textbf{The escape window, geometrically.} In the hybrid sweep, the weight score determines the observed decision because the feature score remains inside its non-triggering interval. At the matched $300$-episode split, $15/20$ weight scores remain above the accusation threshold. At $600$ episodes, all twenty lie between the positive and negative thresholds. At $1200$, all twenty cross the negative threshold. Moderate negation moves the score into this interval, while stronger negation can move it through the interval. The corresponding score ranges are $82$--$155$, $-6$--$73$, and $-236$--$-163$, respectively. No innocent is flagged in the sweep. The attacker does not observe the biases or exact score in this experiment, but may estimate useful attack budgets by other means. The sweep does not establish resistance to adaptive budget selection.

\subsection{The Feature-Space Carrier}
\label{label_method_feature}

\begin{definition}[Feature-space margin carrier]
\label{def:feature-carrier}
For a backbone $f_\theta:\mathcal X\to\mathbb R^D$, fixed probes $p_b$, and unit directions $v_b\in\mathbb S^{D-1}$, define
\begin{equation}
\label{eq_feature_bit}
m_b(\theta)=\inner{v_b}{f_\theta(p_b)},\qquad y_b=\ind\{m_b(\theta)>0\}.
\end{equation}
Embed row $X_i$ by minimizing
\begin{equation}
\label{eq_feature_hinge}
\mathcal L_i^f=\mathcal L_{\mathrm{task}}+\frac{\lambda_t}{\nt}
\sum_b\max\{0,\mu-(2X_{i,b}-1)m_b(\theta)\}.
\end{equation}
\end{definition}

Here $D=512$, smaller than $\omega=4800$. The ability to fit many probe constraints comes from training a nonlinear function on different inputs, not from feature width exceeding the scale dimension. The experiments demonstrate that the carrier can fit $2048$ or $4096$ probes.

\begin{remark}[Carrier-independent scoring]
\label{rem:instantiation}
The conditional soundness bound uses only the decoded bits, biases, and innocent-row independence, so it applies to both readouts. Completeness additionally depends on the attack and channel. Inverting the asymptotic design equation gives the planning index
\begin{equation}
\label{eq:c-from-carrier}
c_{\mathrm{plan}}=(1-2q)\sqrt{\nt/(d\ln(N/\varepsilon_1))}.
\end{equation}
It is not a finite certified coalition size. Theorem~\ref{thm:finite-completeness} provides a finite channel analysis instead.
\end{remark}

Feature matching optimizes
\begin{equation}
\label{eq:kd-obj}
\min_{\theta_S}\E_{x\sim\mathcal D}\|f_S(x)-f_T(x)\|_2^2.
\end{equation}
Its relation to a finite probe set requires control on that set, not merely a small training loss. Proposition~\ref{prop:kd-robust} is stated directly for the probe residual and counts all incorrect or insufficient-margin teacher probes.

A stronger result connects feature stability directly to the weighted tracing score, without independent flips.
\begin{theorem}[Tracing-score stability under arbitrary probe errors]
\label{thm:score-stability}
For client $i$, define $a_{i,b}=2|X_{i,b}-p_b|/\sqrt{p_b(1-p_b)}$. If teacher and student decodes differ on a set $J$, then
\begin{equation}
\label{eq:score-stability}
|S_i^S-S_i^T|\le\sum_{b\in J}a_{i,b}.
\end{equation}
If at most $k$ bits differ, the right side is bounded by the sum $A_i(k)$ of the $k$ largest $a_{i,b}$. Consequently $S_i^T-A_i(k)>z$ guarantees that client $i$ still exceeds threshold $z$, regardless of error dependence. For feature readouts, let $r_b=\|f_S(p_b)-f_T(p_b)\|_2$ and $J_\mu=\{b:|m_b(\theta_T)|<\mu\}$. One may take
\begin{equation}
\label{eq:score-k}
k=\min\{\nt,\ |J_\mu|+\lfloor\nt\eKD^2/\mu^2\rfloor\}.
\end{equation}
Alternatively the directly measured set $\{b:r_b\ge|m_b(\theta_T)|\}$ supplies a sharper bound.
\end{theorem}

This yields a sufficient tracing-survival certificate for correlated or targeted errors. It certifies threshold survival, while Theorem~\ref{thm:tardos} separately certifies the false-accusation risk of the resulting decode. When the bound is too loose, direct score evaluation is still possible. Function-only distillation can rotate feature coordinates, so it need not obey a small probe-residual condition. The observed $20/20$ feature-matching and $19/20$ cross-architecture outcomes, and $0/20$ function-only outcomes, remain empirical evidence rather than substitutes for a margin check.

\textbf{Carrier tradeoff.} The weight carrier is inexpensive to read and has a reported episodic copy cost of about $4.5$ points in the $300$-episode, $\lambda_t=6$ configuration, but distillation removes it. The feature carrier costs $4.8$ and $6.1$ accuracy points per copy on GNSS and CIFAR-10, and supports the longer tracing rows used in the experiments. The design targets are two and three colluders. Measured tracing extends beyond those targets. The score-stability theorem explains a sufficient mechanism for survival without claiming that feature matching always preserves the mark.

\section{Assumptions and Supporting Lemmas}
\label{app:theory-detail}
\subsection{Assumptions}
\label{sub:assumptions}

\begin{assumption}[Unconditioned codebook geometry]
\label{ass:geom}
Identity columns are independent uniform directions on $\mathbb S^{\omega-1}$. The implementation normalizes one Gaussian draw without a global rejection test.
\end{assumption}
\begin{assumption}[Achieved local margins]
\label{ass:margin}
At the round analyzed, each participating client satisfies $t_{i,b}\inner{\gamma^{(i)}}{E_{i,b}}\ge\mu$ for every bit. A hinge objective alone does not establish this premise, so the achieved margins must be checked or treated as idealized.
\end{assumption}
\begin{assumption}[Conditional signed-direction model]
\label{ass:indep}
Conditional on $(u_i,t_i)$, the signed directions $\{t_{i,b}E_{i,b}\}_b$ remain independent uniform sphere directions. This is an idealized single-round model. Repeated federated training is not asserted to satisfy it.
\end{assumption}
\begin{assumption}[Identity codeword model]
\label{ass:codewords}
Identity codewords are independent uniform binary vectors. They need not be independent of the trained decoded aggregate.
\end{assumption}
\begin{assumption}[Tracing probability models]
\label{ass:leak}
(a) Conditional on the recovered word and biases, each innocent row retains independent $\mathrm{Bernoulli}(p_b)$ entries. This holds when the artifact and its selection expose no information about that row beyond the biases. (b) For the finite completeness theorem only, the coalition's side information is independent of biases conditional on its own rows, its intended word satisfies marking, and residual flips are independent of all rows and biases and mutually independent with common rate $q$. Part (a) does not require part (b).
\end{assumption}
\begin{assumption}[Comparable feature coordinates]
\label{ass:kd}
Teacher and student features use the same coordinates and unit projection vectors. The residual is evaluated on the actual fixed probes. Inferring that residual from a population or training objective needs a separate generalization argument.
\end{assumption}

A3 is not derived from the fact that clients train on separate blocks: earlier aggregates carry other clients' marks. Similarly, uniform model averaging does not prove the BSC premise in A5(b). Innocent-score soundness, finite BSC completeness, and deterministic score stability have different assumptions and should be applied separately.

\subsection{Projection and Geometry}
\label{sub:sphere}
The match score is
\begin{equation}
\label{eq:match}
M(i,j)=1-\HD(\hat h_i,w_j)/n.
\end{equation}
The projection decomposition is exact:
\begin{equation}
\label{eq:split}
\inner{\gagg}{E_{i,b}}=\lambda_i\inner{\gamma^{(i)}}{E_{i,b}}+\inner{u_i}{E_{i,b}}.
\end{equation}
For $\gamma^{(k)}=\gzero+\Delta_k$, the norm inequality is
\begin{equation}
\label{eq:unorm}
\|u_i\|\le(1-\lambda_i)G+\sum_{k\ne i}\lambda_k\|\Delta_k\|.
\end{equation}
The quadrature approximation \eqref{eq:snr-scaling} drops baseline-update and update-update cross terms. It is not implied by A1 or by a Gram norm below $2\sqrt\rho+\rho$.

\begin{lemma}[Rank and random-code separation]
\label{lem:capacity}
If $N_cn>\omega$, the identity Gram matrix is singular and $\|E^\top E-I\|_{\mathrm{op}}\ge1$. This does not preclude useful sign decoding. Under A4, for $0<\zeta<1/2$,
\begin{equation}
\label{eq:near-collision}
\Prob[\exists i<j:\HD(w_i,w_j)\le(1/2-\zeta)n]
\le\binom{N_c}{2}e^{-2\zeta^2n}.
\end{equation}
\end{lemma}

For instance, $E=[I\ I]$ has load two and Gram deviation one, and passes the inequality with right side $2\sqrt\rho+\rho$. Thus $N^*=\omega/n=37.5$ is a dimensional reference, not an impossibility theorem. For unit-normalized Gaussian columns the limiting nonzero singular-value edges use $1\pm\sqrt\rho$, without another division by $\sqrt\omega$~\cite{vershynin2018hdp}. The empirical failure point also depends on codeword separation and decoding errors.

\section{Proofs}
\label{app:proofs}
\subsection{Identity Recovery and Attribution}
\label{sub:proofs-survival}

\begin{lemma}[Sphere projection]
\label{lem:sphere}
For $X$ uniform on $\mathbb S^{\omega-1}$ and fixed $v\ne0$, the normalized projection $\sqrt\omega\inner{X}{v}/\|v\|$ is symmetric with variance one and, for $\omega>4$,
\begin{equation}
\label{eq:be}
\sup_a|F_\omega(a)-\Phi(a)|\le8/(\omega-4).
\end{equation}
\end{lemma}
\begin{proof}
Rotation reduces the projection to the first coordinate. Symmetry follows by reflection. The variance is $1/\omega$ before normalization because the squared coordinates sum to one and have equal expectations. The first-coordinate normal approximation follows from the finite-sphere bound of Diaconis--Freedman~\cite{diaconis1987finetti}. For $v=0$ the unnormalized projection is identically zero and normalization is unnecessary.
\end{proof}

\begin{proof}[Proof of Theorem~\ref{thm:perbit}]
A2 and \eqref{eq:split} imply \eqref{eq:identity-proxy-main}. Conditional on $(u_i,t_i)$, A3 makes the signed noise projection a fixed vector projected onto an independent random unit direction. Its symmetry gives proxy success probability $F_\omega(s_i)$. The true success event contains the proxy event, so it has at least that probability. Different proxy bits use independent signed directions. Lemma~\ref{lem:sphere} supplies the normal approximation. If $u_i=0$, the local positive margin survives multiplication by $\lambda_i>0$.
\end{proof}

\begin{proof}[Proof of Theorem~\ref{thm:attr}]
For any rival $j$, the triangle inequality gives $\HD(\hat h_i,w_j)\ge\HD(w_i,w_j)-d_i$. Thus $2d_i<d_{\min}$ makes every rival farther away than the true codeword. For the probabilistic statement, an error requires either $\max_i d_i>rn$ or a pair of codewords at distance at most $2rn$. Under A4, each pair distance is $\mathrm{Binomial}(n,1/2)$, so Hoeffding's inequality bounds the latter event by the second term of \eqref{eq:attr-thresh}. The union bound requires no independence between decoding errors and codewords.
\end{proof}

\begin{proof}[Proof of Corollary~\ref{cor:envelope}]
The true error count is bounded above by the number of failed proxy bits. Conditional on $(u_i,t_i)$, these bits are independent with mean at most $1-p$. Hoeffding gives $\Prob[d_i>(1-p+\xi)n]\le e^{-2n\xi^2}$ after removing the conditioning. Union over clients and apply Theorem~\ref{thm:attr}. This gives \eqref{eq:master}.
\end{proof}

\begin{proof}[Proof of Lemma~\ref{lem:capacity}]
More columns than rows imply a zero eigenvalue of $E^\top E$, hence an eigenvalue $-1$ of $E^\top E-I$. This proves only the rank assertion. Each independent codeword-pair distance is binomial. Applying the lower Hoeffding tail and taking a union over pairs proves \eqref{eq:near-collision}.
\end{proof}

\begin{proposition}[Exact presence calibration]
\label{prop:presence-honest}
Against an independent uniform binary codeword, a Hamming-radius-$t$ test has false-accept probability
\begin{equation}
\label{eq:presence-band}
P_{\mathrm{FA}}(n,t)=2^{-n}\sum_{j=0}^t\binom nj.
\end{equation}
At $n=128$, the largest radius satisfying $P_{\mathrm{FA}}\le2^{-128}$ is zero. Under an additional i.i.d. true-bit model with accuracy $p$, completeness is
\begin{equation}
\label{eq:presence-gap}
P_{\mathrm{accept}}=\sum_{j=0}^t\binom nj(1-p)^j p^{n-j}.
\end{equation}
For completeness at least $.95$, exact enumeration gives minimum lengths $202$ at $p=.95$ ($t=15$) and $233$ at $p=.93$ ($t=23$).
\end{proposition}
\begin{proof}
Count the binary vectors in the Hamming ball for the false-accept law. Under the separate i.i.d. error model the true Hamming distance is binomial. At $n=128$, radius zero contains one vector and radius one contains $129$, proving the exact calibration. Evaluating the two binomial tails jointly gives the stated example lengths. At $n=128$, perfect extraction has completeness one, while $p=.9999$ gives $.9873$. Dependent decoded bits require their own completeness analysis. The false-accept calculation concerns an independent codeword, not deliberate copying of a public mark.
\end{proof}

\subsection{Tracing and Feature-Stability Proofs}
\label{app:tardos-proofs}
\begin{proof}[Proof of Lemma~\ref{lem:innoc} and Theorem~\ref{thm:tardos}]
For $X\sim\mathrm{Bernoulli}(p)$ independent of the decoded bit, $\E[X-p]=0$ and $\E[(X-p)^2]=p(1-p)$. Substitution into \eqref{eq:U-def} gives zero mean and unit second moment per increment. A5(a) gives independence across positions, so the score variance is $\nt$. The two-point MGF is exactly \eqref{eq:mgf-chernoff}. Markov's inequality applied to $e^{\alpha S_i}$ and a union over innocent rows prove \eqref{eq:tardos-sound}. For the uniform alternative, Bernstein's inequality gives
\[
\Prob[S_i>z\mid y,p]\le\exp\{-z^2/(2\nt+2Bz/3)\}.
\]
Solving $z^2/(2\nt+2Bz/3)=L$ yields \eqref{eq:tardos-sound-proved}. Equal moments do not imply equal MGFs: at $p=.1$, the laws for $y=0$ and $y=1$ are reflected asymmetric two-point distributions.
\end{proof}

\begin{proof}[Proof of Theorem~\ref{thm:finite-completeness}]
Condition on all coalition rows. Biases remain independent across positions under this conditioning. A position containing $k$ ones has posterior bias density proportional to $f_{\delta_c}(p)p^k(1-p)^{s-k}$. Conditional on the rows and an intended output word, the independent BSC flips give the exponential factor in \eqref{eq:finite-coal-mgf}. A randomized strategy is a mixture of such words. Its conditional product of factors is at most the product of the largest allowed factor at each position, even when it chooses its symbols jointly. Now average over the independent coalition columns. The posterior normalization cancels the column probability, and summing over the $\binom{s}{k}$ columns with $k$ ones gives $J_s(t,q)$ per position. Hence $\E[e^{-t\sum_{i\in\Coll}S_i}]\le J_s(t,q)^{\nt}$. If all colluder scores are at most $z$, their sum is at most $sz$. Exponential Markov bounds that event by \eqref{eq:finite-completeness}.

For soundness with \eqref{eq:finite-innoc-mgf}, first condition on biases and the output. Each innocent MGF is bounded by the maximum over the two output symbols. The resulting product depends only on the biases, whose independent draws give $I(a)^{\nt}$ after averaging. Markov and a union over at most $N$ innocents give $N e^{-az}I(a)^{\nt}$. This is an a-priori guarantee over code generation. It does not condition on selecting favorable realized codebooks.
\end{proof}

\begin{proof}[Proof of Lemmas~\ref{lem:coalmean} and~\ref{lem:flip}]
For an independent sign flip, $\E[1-2F]=1-2q$ and $(1-2F)^2=1$. Subtracting the squared new mean from the unchanged second moment gives the variance increase $4q(1-q)(\E[T])^2$. For a unit carrier direction, $|\inner{\eta}{E_{T,b}}|\le\|\eta\|_2$, so a larger absolute noiseless margin cannot change sign.
\end{proof}

\begin{proof}[Proof of Corollary~\ref{cor:secret-row}, Proposition~\ref{prop:adjudicate}, and Lemma~\ref{lem:tamper-feature}]
Antisymmetry makes the score of the complemented decode equal $-S_i$, allowing the same upper-tail certificate to treat negative scores. For each fixed decode, a passing check bounds the allocated tail event. A failed check prevents a certified decision. Union over the $2K$ allocated events proves the total budget without requiring carrier independence. Precedence can only remove decisions. For perfect negation, direct substitution into the score gives the negative absolute increment, whose expectation is $-2\sqrt{p_b(1-p_b)}$.
\end{proof}

\begin{proof}[Proof of Proposition~\ref{prop:kd-robust} and Theorem~\ref{thm:score-stability}]
Cauchy--Schwarz bounds the probe-margin change by $r_b=\|f_S(p_b)-f_T(p_b)\|_2$. A correctly embedded teacher margin at least $\mu$ is preserved if $r_b<\mu$. Charge the $q_{\mathrm{bad}}\nt$ remaining probes in full and use $\#\{b:r_b\ge\mu\}\mu^2\le\sum_b r_b^2$ to prove \eqref{eq:qkd}. This is a deterministic finite-sample inequality.

Changing one binary output reverses its score increment, changing $S_i$ by absolute amount $a_{i,b}$. Summing over changed positions proves \eqref{eq:score-stability}. Maximizing a sum of $k$ nonnegative weights selects the largest $k$. Only probes with small teacher margins or $r_b\ge\mu$ can change, giving \eqref{eq:score-k}. The direct residual-to-margin comparison is valid without in-distribution sampling or independence of the errors. It does not imply that the upper bound is tight.
\end{proof}

\subsection{Credential Proof and Security Accounting}
\label{app:prop-security}
\label{sub:soundness-layers}

\begin{assumption}[Random-oracle and commitment model]
\label{ass:rom}\label{ass:commit}
Fiat--Shamir is modeled with a classical random oracle. Commitments use a fresh random salt and SHAKE-256 with a 256-bit output, modeled as hiding and binding. The ideal query-limited collision bound is $\binom Q2 2^{-256}$. An extraction reduction must account for all of its oracle queries when applying this bound.
\end{assumption}
\begin{assumption}[Computational witness recovery]
\label{ass:xlpn}
Recovering a weight-$w_\tau$ error $e$ with $y\oplus e\in\Img(A)$ from a random registered instance is assumed computationally hard. The experimental parameters are $m=1024$, $l=512$, $\tau=.125$, $w_\tau=128$. Information-set-decoding estimates quantify computational work~\cite{esser2022syndrome}, not forgery probabilities.
\end{assumption}

\begin{definition}[Credential knowledge relation]
\label{def:rel}
The relation is $\mathcal R=\{((A,y),e):\wt(e)=w_\tau,\ y\oplus e\in\Img(A)\}$. A witness $(s,e)$ with $y=As\oplus e$ satisfies it.
\end{definition}

\begin{lemma}[Three-transcript extraction]
\label{lem:special}
Three accepting Stern transcripts with the same commitments and all three challenges reveal a witness for $\mathcal R$, provided the commitments bind and the encoded permutation is valid.
\end{lemma}
\begin{proof}
The openings jointly determine $\pi,t_0,t_1,t_2$. The first two checks imply $t_0\oplus\pi^{-1}(t_1)\in\Img(A)$ and $t_0\oplus\pi^{-1}(t_2)\oplus y\in\Img(A)$. XOR gives $y\oplus\pi^{-1}(t_1\oplus t_2)\in\Img(A)$. The third check and permutation invariance give weight $w_\tau$, so $e'=\pi^{-1}(t_1\oplus t_2)$ is a witness. This is the Stern-type extraction used in the xLPN protocol~\cite{jain_krenn_pietrzak,veron,yang_yin_zhu}.
\end{proof}

\begin{lemma}[Ideal Fiat--Shamir knowledge error]
\label{lem:afk}
With ideal binding commitments and uniformly sampled ternary challenges, the $r$-fold parallel protocol has knowledge error $(2/3)^r$. Its single-challenge-phase Fiat--Shamir transform has knowledge error at most $(Q+1)(2/3)^r$ under generalized special-soundness extraction~\cite{attema2026generalized}.
\end{lemma}
\begin{proof}
Let $\Gamma$ contain challenge sets exhibiting all three values in some coordinate. Such a set extracts by Lemma~\ref{lem:special}. A non-extracting set has at most $2^r$ vectors, giving $\kappa_\Gamma=(2/3)^r$. Each useful challenge adds an unseen coordinate value, so $t_\Gamma\le2r+1$. Useful challenges can be sampled by rejection outside the Cartesian product of previously seen values. Until extraction its probability is at most $(2/3)^r$. Theorem~5 of~\cite{attema2026generalized} therefore applies with polynomial $T_\Gamma\le2r+2$. Its extractor uses at most $(Q+1)(2r+2)/(1-\kappa_\Gamma)$ expected prover calls and succeeds with probability at least $(\epsilon-(Q+1)\kappa_\Gamma)/(1-\kappa_\Gamma)$. The linear bound on useful challenges establishes efficient knowledge extraction. Commitment failures must be added at the reduction's actual query budget.
\end{proof}

\begin{corollary}[Round-count calibration]
\label{cor:rounds}
For the ideal knowledge-error term to be at most $2^{-129}$, it suffices that
\begin{equation}
\label{eq:rderiv}
r\ge\frac{129+\log_2(Q+1)}{\log_2(3/2)}.
\end{equation}
At $Q=2^{64}$, $330$ rounds suffice and the implementation uses $331$. The ideal term is approximately $2^{-129.623}$.
\end{corollary}
\begin{proof}
Take logarithms of $(Q+1)(2/3)^r\le2^{-129}$ and round upward. Here $\log_2(2^{64}+1)>64$. Keeping that term does not change the integer result. The experimental mapping of 16-bit words modulo three has maximum two-answer probability $43691/65536$, giving approximately $2^{-129.619}$ instead. The verifier uses this mapping, with its bias included in the calibration. Rejection sampling provides an exactly uniform alternative.
\end{proof}

\begin{proof}[Proof of Proposition~\ref{prop:security}]
\emph{Completeness.} Honest openings satisfy the three algebraic checks for their respective challenges. If the required presence condition also holds, verification accepts. This is conditional completeness. Noisy extraction can fail presence even for an honest credential holder.

\emph{Knowledge and zero knowledge.} The extraction statement follows from Lemmas~\ref{lem:special} and~\ref{lem:afk} in the ideal model. A simulator chooses challenges first, samples the corresponding accepting masked openings, commits arbitrary hidden values in unopened slots, and programs the challenge oracle at the complete commitment/context input. Hiding and fresh commitment entropy bound the distinguishing effects of unopened values and prior oracle queries. This yields the usual computational random-oracle simulation, subject to the commitment assumptions~\cite{damgard,attema2022fiatshamir}.

\emph{Binding.} The challenge context contains the extracted component, registered credential, and a canonical digest of the complete model state. An unchanged transcript then transfers to different model bytes only through a digest collision or a new-context challenge coincidence. Binding only the extracted component would not distinguish models with equal extracted bits. Model-state binding authenticates the credential statement rather than authorship: a public watermark can be copied, and a legitimate witness holder can authenticate after that copying.
\end{proof}

\begin{table}[t]
\centering\footnotesize
\caption{Security quantities and their separate meanings.}
\label{tab:soundness-split}
\begin{tabular}{p{.36\columnwidth}p{.52\columnwidth}}
\toprule
Quantity & Interpretation\\\midrule
Ideal FS knowledge error & $(Q+1)(2/3)^r$; extraction statement\\
Experimental ternary mapping & replace $2/3$ by $43691/65536$ for grinding calibration\\
Commitment failure & bounded at the reduction's total query budget\\
xLPN recovery work & estimated attack cost; not $2^{-128}$ forgery probability\\
Tracing false accusation & certified tail budget under A5(a)\\
Tracing completeness & finite bound \eqref{eq:finite-completeness} under A5(b)\\
\bottomrule
\end{tabular}
\end{table}

An end-to-end credential-forgery reduction must additionally bound witness-recovery advantage and extraction cost. The implementation validates encodings, uses operating-system randomness for proof masks and permutations, and binds transcripts to the full model state. Tracing-code generation uses a domain-separated SHAKE-256 stream with a 32-byte secret key. The experiments use fixed codebooks generated from a fixed experimental key and 64-bit NumPy seeds for reproducibility. These reproducible codebooks support experimental comparisons. Operational secrecy requires secret cryptographic generation, authenticated decoder setup, and the row-independence premise.

\section{Extended Results}
\label{app:extended-results}

\subsection{Configuration Details}
\label{sub:config}

The ten fixed seeds behind every reported mean~$\pm$~standard deviation are $\{42, 137, 271, 314, 1729, 2718, 3141, 5772, 6561, 9999\}$. Collusion traceability averages $20$ random coalitions' copies per size for the weight carrier and $50$ for the feature carrier, pooled over the ten seeds. Both training phases of Sec.~\ref{label_experiments} use SGD with momentum $0.9$, weight decay $5{\times}10^{-4}$, and cosine annealing, at learning rate $0.01$ for the 200-epoch cross-entropy pre-training and $\eta{=}0.001$ annealed over the $R$ federated rounds.

\textbf{Deployment scope of the simulated federation.} All stations use partitions of one GNSS recording campaign~\cite{heublein_feigl_crpa}. The Dirichlet partition varies class coverage and sample count, but does not measure physical differences between sites. We sweep label skew to $\alpha{=}0.1$ and participation to $30\%$ per round (Table~\ref{tab:noniid}). Extreme skew leaves some clients unable to form a five-way episode. Partial participation models absence for an entire round. The receiver-shift experiment adds a $+3$\,dB gain offset and per-station SNRs of $5$--$20$\,dB (App.~\ref{sub:stressors}).

These tests do not cover distinct antenna and front-end calibrations, independent multipath, or local interference at separately sited receivers. The shift sweep is a proxy for receiver variation. Compute heterogeneity is also untested: all clients use the same local-episode budget, so the evidence does not cover stragglers, mid-round dropout, or unequal training progress.

Physical deployment may change extraction errors, the required length in \eqref{eq:design}, and tracing completeness. Each deployment decode requires its own interval certificate (Remark~\ref{rem:deployed-vs-provable}). Under A5(a), innocent scores have mean zero and variance $\nt$, but their tail bounds depend on the recovered bits and biases. In the proxy experiment, tracing persists in all $100$ trials while exact single-leaker isolation degrades toward chance.

\subsection{Attribution under Non-IID Data and Partial Participation}
\label{sub:noniid}

Table~\ref{tab:noniid} separates attribution over the full roster from attribution over clients that can train. With $30\%$ participation per round, roster-wide attribution remains $98\%$ on GNSS and $100\%$ on CIFAR-10. Under strong label skew, some clients have too few classes for a five-way episode. They do not embed a mark, even though no shard is empty. At $\alpha{=}0.5$, GNSS roster-wide attribution is $79\%$. Participating-only attribution is $100\%$ in every reported cell on both datasets. The sweep therefore identifies episode formation as the source of the roster-wide losses in these configurations.

\renewcommand{\arraystretch}{0.95}
\begin{table}[t]
\setlength{\tabcolsep}{4pt}
    \centering
    \caption{Attribution and few-shot accuracy under label non-IID (Dirichlet concentration $\alpha$) and partial participation (per-round client fraction $q$), $N_c{=}10$, 10-seed mean$\pm$std, GNSS first. \emph{Raw} denotes attribution over all ten clients. Participated-only attribution is $100\%$ in every cell (footnote). $^{\ast}$~The deployed configuration is evaluated in an independent run of the sweep. Its GNSS accuracy sits within one standard deviation of Table~\ref{tab:attribution}, whose wider-spread GNSS partition draws are the source of the difference.}
    \label{tab:noniid}
    \begin{tabular}{lcccc}
    \toprule
    & \multicolumn{2}{c}{\textbf{GNSS}} & \multicolumn{2}{c}{\textbf{CIFAR-10}} \\
    \cmidrule(lr){2-3}\cmidrule(lr){4-5}
    \textbf{Setting} & \textbf{Attr.\,(\%)} & \textbf{Acc\,(\%)} & \textbf{Attr.\,(\%)} & \textbf{Acc\,(\%)} \\
    \midrule
    $\alpha{=}0.1,\ q{=}1.0$ & $16.0_{\pm12.0}$ & $58.2_{\pm9.4}$ & $51.0_{\pm14.5}$ & $74.6_{\pm4.5}$ \\
    $\alpha{=}0.25,\ q{=}1.0$ & $32.0_{\pm16.6}$ & $84.1_{\pm8.3}$ & $91.0_{\pm7.0}$ & $81.7_{\pm0.6}$ \\
    $\alpha{=}0.5,\ q{=}1.0$ & $79.0_{\pm13.0}$ & $92.0_{\pm0.9}$ & $100.0_{\pm0.0}$ & $83.3_{\pm0.4}$ \\
    $\alpha{=}2.0,\ q{=}1.0^{\ast}$ & $98.0_{\pm4.0}$ & $92.3_{\pm2.6}$ & $100.0_{\pm0.0}$ & $84.9_{\pm0.2}$ \\
    $\alpha{=}2.0,\ q{=}0.5$ & $98.0_{\pm4.0}$ & $91.4_{\pm2.2}$ & $100.0_{\pm0.0}$ & $84.5_{\pm0.3}$ \\
    $\alpha{=}2.0,\ q{=}0.3$ & $98.0_{\pm4.0}$ & $90.7_{\pm2.2}$ & $100.0_{\pm0.0}$ & $84.0_{\pm0.4}$ \\
    \bottomrule
    \end{tabular}

    \vspace{2pt}
    {\footnotesize Participated-only attribution (restricted to clients that form at least one episode) is $100\%$ in every cell on both datasets, and no shard is ever empty. Under $\alpha{=}0.1$ a mean of $4.1$ (CIFAR-10) and $0.5$ (GNSS) of ten clients can assemble a $5$-way $5$-shot episode, so the raw rate tracks episode-formability rather than attribution loss. On GNSS no client forms one in six of the ten seeds, leaving that cell's participated-only figure resting on the remaining four.}
\end{table}

\subsection{Tracing, Collusion, and Anti-Framing}
\label{sub:tracing-ext}

\textbf{Collusion for removal, extended settings.} Section~\ref{sec:eval_tracing} reports traceability against coalition size $k$ at the operating point. Two extensions complete the picture. The weight-space overlay repeats its profile at $N_c{=}40$, tracing all $k{=}2$ and $k{=}3$ coalitions ($10/10$) and $0.70$ of $k{=}5$ on GNSS ($0.20$ on CIFAR-10), again with zero framing. In the adaptive-steering experiment, the simulator edits the decoded word toward the coalition majority. It does not optimize model parameters toward a chosen innocent. At $k{=}2$, every tested steered coalition is traced. Stronger steering lowers the colluders' own flip rate from $0.28$ to $0.11$, increasing traceability in these configurations. No innocent is framed across $400$ steered trials. These observations characterize the tested attack. The formal soundness guarantee remains subject to Theorem~\ref{thm:tardos} and Assumption~\ref{ass:leak}. The mean disagreement with individual colluder rows $\bar q$ on the weight carrier rises with $k$, measuring $0.22/0.24/0.26/0.26$ (GNSS) and $0.25/0.26/0.27/0.28$ (CIFAR-10) at $k{=}2/3/5/8$.

\textbf{Cooperation-free tracing and registry adjudication.} Tracing an unknown-origin model needs no help from the leaker. Rebuilding the carriers from the public seeds and decoding the copy offline names the exact leaker in all ten seeds ($10/10$),\footnote{The released per-seed logs record these two outcomes under the metric identifiers \texttt{trace\_exact} and \texttt{frame\_tardos\_hit}.} and the registry adjudicates the accusation against the enrolled commitments in all ten ($10/10$), so a disputed model is resolved without the suspect ever participating. This closes the gap left by proof-of-ownership schemes that confirm only a cooperating owner.

\textbf{Anti-framing and credential verification.} Copying a victim's public identity codeword raises its self bit-accuracy to $1.00$, showing that public-codeword presence alone can be forged. The secret overlay accuses the framed victim in $0/10$ trials. The credential proof accepts the legitimate victim in $10/10$ trials and accepts the tested forgeries in $0/10$. Registry adjudication also rejects evidence substitution in all ten trials.

The credential acceptance and forgery benchmarks use transcripts bound to the extracted component and public credential. To evaluate full-state binding, we change the model state while holding the extracted component fixed. The verifier rejects the transferred transcript. These empirical rejection counts are distinct from the proof knowledge-error bound below $2^{-128}$ at $r=331$ under the stated query budget (Proposition~\ref{prop:security}). Wrong-client, server-side, and replay attempts are included in the forgery tests. A credential proof costs approximately $0.6$\,MB and $37$\,s per client. The repetitions can be verified in parallel. This is an offline dispute procedure, not part of each training round.

\textbf{Dispatch without client data.} The two carriers share every downstream step, the code, the registry, credential authentication, and the conditional score bound (Remark~\ref{rem:instantiation}, Theorem~\ref{thm:tardos}). Section~\ref{sec:eval_tracing} sets their operational trade side by side. Neither requires the operator to hold client data. Writing the feature carrier from a server-held proxy set reproduces the same profile, $100/100$ isolation and $10/10$ distillation survival at a $4$--$5$-point copy cost, with tracing unaffected by adaptive steering. The deployed global model is untouched by either carrier, since both are written only into dispatched copies (App.~\ref{label_method_dispatch}). The global-accuracy cost is zero.

\subsection{Distillation Leveling}
\label{sub:kd-leveling}

\textbf{The single cross-architecture failure.} ResNet-18-to-ResNet-34 distillation preserves the feature mark in $19/20$ runs. The failed run is GNSS seed $271$: its post-distillation flip rate is $0.4976$ and its score is $54.8$. Its same-architecture flip rate was $0.147$, below the GNSS median. The increase of $0.351$ is the largest among all twenty runs, so the failure reflects a large architecture-induced change rather than a mark already close to chance.

The sweep holds the probes, code, and $80$-epoch budget fixed. ResNet-34 is chosen because its penultimate layer has the same width, $512$, as the teacher's. On GNSS, the mean flip rate increases from $0.195$ under same-architecture distillation to $0.291$ under cross-architecture distillation, a rise of $0.096$. On CIFAR-10 it increases from $0.119$ to $0.165$, a rise of $0.046$. Cross-architecture rates also span a wider range on GNSS ($0.165$--$0.498$) than on CIFAR-10 ($0.132$--$0.211$). The feature carrier's clean flip rates are below $0.09$ and $0.002$, respectively.

Equation~\eqref{eq:qkd} provides an interpretation: feature-matching errors that exceed the teacher's probe margins can flip decoded bits. It bounds the fraction of potentially changed bits. Theorem~\ref{thm:score-stability} translates those changes into a score bound by weighting each position by its score impact. The measured flip rates are outcomes, however, and do not directly measure $\eKD$ or establish that this upper bound is tight. GNSS's sparse four-channel inputs and subtler inter-class differences are plausible sources of transfer difficulty, but the experiment does not isolate them. The larger per-class recall cost on GNSS ($10.9$ versus $4.4$ points) is consistent with a more demanding representation problem. The baseline scale $G$ appears only in the weight-carrier bound and does not explain this feature-carrier failure.

The failed run loses the mark without a false accusation. Across the distillation evaluation, no innocent is framed in the $180$ runs returning a framing verdict. The remaining twenty rows are the DeepMarks-BIBD weight-carrier comparator under cross-architecture transfer, for which the BN-$\gamma$ extraction layout is incompatible with ResNet-34 and no framing verdict is returned.

\textbf{Distillation across the benchmark.} Table~\ref{tab:kd} evaluates each native mark after the same $80$-epoch distillation. FedIPR, FedZKP, and our weight carrier lose their marks on both datasets ($0/10$ per dataset). FedTracker's attribution is at chance ($0.10$). Output-space methods vary by dataset: WAFFLE retains its mark in $3/10$ GNSS runs and $0/10$ CIFAR-10 runs, while DUW retains its per-client key in $3/10$ and $10/10$, respectively. WAFFLE's clean GNSS mark is already near its chance level.

DUW's CIFAR-10 survival shows that a key in the teacher's function can transfer to the student, but DUW lacks a collusion-secure code and certified false-accusation bound. Our feature carrier combines that transfer mechanism with coded tracing. It survives feature matching in $20/20$ runs and cross-architecture transfer in $19/20$. Logit-only distillation and feature isometry erase it ($0/20$ each; Table~\ref{tab:kdclass}). Thus the observed advantage is specific to distillers that reproduce the feature geometry closely enough, as required by Assumption~\ref{ass:kd}.

\renewcommand{\arraystretch}{0.95}
\begin{table}[t]
\setlength{\tabcolsep}{5pt}
    \centering
    \caption{Feature-carrier survival by distiller class, 20 seeds (10 GNSS $+$ 10 CIFAR-10). Survival counts are shown for each tested class. Weight-space marks fail where extraction is defined. The incompatible cross-architecture DeepMarks-BIBD runs return no framing verdict.}
    \label{tab:kdclass}
    \begin{tabular}{llc}
    \toprule
    \textbf{Distiller} & \textbf{Student copies} & \textbf{Survival} \\
    \midrule
    Feature-matching (FitNets)       & feature geometry & \textbf{\textcolor{ForestGreen}{$20/20$}} \\
    Cross-architecture (R18$\to$R34) & feature geometry & \textcolor{ForestGreen}{$19/20$} \\
    Logit-only KD                    & output function  & \textcolor{red}{$0/20$} \\
    Feature isometry                 & output function  & \textcolor{red}{$0/20$} \\
    \bottomrule
    \end{tabular}
\end{table}

\renewcommand{\arraystretch}{0.95}
\begin{table}[t]
\setlength{\tabcolsep}{5pt}
    \centering
    \caption{Distillation leveling: native mark retained after $80$-epoch KD, as a fraction of $10$ seeds surviving (GNSS first). Under this feature-matching protocol, weight-space marks fail, output-space results depend on the dataset, and the feature carrier survives on both datasets.}
    \label{tab:kd}
    \begin{tabular}{llcc}
    \toprule
    \textbf{Method} & \textbf{Carrier} & \textbf{GNSS} & \textbf{CIFAR-10} \\
    \midrule
    \textbf{Ours (feature)} & feature-space & \textbf{\textcolor{ForestGreen}{$10/10$}} & \textbf{\textcolor{ForestGreen}{$10/10$}} \\
    Ours (weight) & weight-space & \textcolor{red}{$0/10^{\ddagger}$} & \textcolor{red}{$0/10^{\ddagger}$} \\
    FedIPR & weight-space & \textcolor{red}{$0/10$} & \textcolor{red}{$0/10$} \\
    FedZKP & weight-space & \textcolor{red}{$0/10$} & \textcolor{red}{$0/10$} \\
    FedTracker & weight-space & $0.10^{\dagger}$ & $0.10^{\dagger}$ \\
    DUW & output-space & $0.30^{\dagger}$ & $1.00$ \\
    WAFFLE & output-space & $0.30$ & $0.00$ \\
    \bottomrule
    \end{tabular}

    \vspace{2pt}
    {\footnotesize $^{\dagger}$~Consistent with per-seed Bernoulli outcomes at the $1/N_c{=}0.1$ chance level (upper-tail $p{=}0.65$ for FedTracker, $p{=}0.07$ for DUW on GNSS). DUW on CIFAR-10 is genuine survival, not chance ($10/10$, $p{\approx}10^{-10}$ under the chance model), and is discussed in the text. $^{\ddagger}$~Retention criterion for our weight carrier: self bit-accuracy, which collapses to its $0.5$ chance level in every seed. Under a twice-chance attribution criterion (attribution ${>}0.2$ against chance $0.1$) the cells read $0/10$ on GNSS and $2/10$ on CIFAR-10.}
\end{table}

\subsection{Insider Attacks and Deployment Stressors}
\label{sub:stressors}

\textbf{Own-row erasure: scores and budgets.} Negated erasure moves the feature score from approximately $+1317$ to $-1265$. The two-sided tamper test ($Z_{\mathrm{two}}{=}201.4$) detects it in all twenty seeds, at a task-accuracy cost of three to four points. On the weight overlay, the same strategy leaves only an uncertified forensic lead. Fresh erasure has the complementary effect: it drives the feature score near zero, while the weight overlay still traces the true row in all twenty runs. Neither strategy frames an innocent.

The hybrid combines weight negation with fresh feature erasure. It escapes the threshold-only tests on both carriers in $5/20$ runs at matched compute, $20/20$ at twice that budget, and $0/20$ at four times the budget, where the weight score crosses the negative tamper threshold. The budgets are $300$, $600$, and $1200$ episodes, with twenty runs per level. Appendix~\ref{ssub:adjudicate} gives the score ranges and geometric explanation. These observations characterize the tested budgets rather than establish an upper bound on an adaptive hybrid's success. Framing remains $0/20$ at each level.

\textbf{Robust aggregation.} With coordinate-wise median aggregation, attribution is $1.00$ on CIFAR-10 and $0.98$ on GNSS, matching the reported FedAvg rates. Self bit-accuracies are $0.92$ and $0.95$, respectively. Task accuracy falls to $80.2\%$ on CIFAR-10 and $63.4\%$ on GNSS, with high GNSS variance. These results show that the tested median configuration retains identity attribution. They do not establish invariance to arbitrary robust aggregators.

\textbf{Receiver covariate shift (GNSS).} We apply a $+3$\,dB gain offset and per-station SNRs from $5$ to $20$\,dB to model receiver variation. Tracing succeeds in all $100$ shifted trials, and aggregate attribution remains $0.98$. Exact single-leaker isolation degrades toward chance under the same perturbations. Thus retaining a tracing accusation does not imply retaining the stronger isolation outcome.

\textbf{Per-class cost on dispatched copies.} We evaluate class recall against a clean control over ten seeds in the original label space. Mean recall falls by $4.4$ points on CIFAR-10 ($0.844\to0.799$), approximately uniformly across classes, and by $10.9$ points on GNSS ($0.932\to0.823$). GNSS losses range from $0.9$ to $18.0$ points, with the largest losses on the harder interference types. These class-level differences provide information that episode-averaged accuracy does not capture.

The benchmark excludes the interference-free class, so it cannot measure false alarms on clean signals. The reported \emph{far\_proxy} is instead inter-class confusion, $P(\mathrm{pred}\ne\mathrm{true}\mid\mathrm{true}=c)$. Open-set abstention is not evaluated. These costs concern the dispatched copies, not the global model tested in App.~\ref{sub:utility-ext}.

\textbf{Embedding channel.} The weight-space carrier is held almost entirely by the BN~$\gamma$ scales, whose standard deviation broadens by $6.7$--$7.1\times$ while $\beta$ is nearly untouched, consistent with a hinge loss that steers each projection to the correct sign with margin. The representation changes are consistent with, but do not isolate, the effect of the identity hinge on the feature-space geometry reported in Section~\ref{sec:eval_robustness}.

\subsection{Global-Model Utility Equivalence}
\label{sub:utility-ext}

\begin{table}[t]
\centering\footnotesize
\caption{Paired global-model equivalence at $N_c=10$, ten seeds. $\Delta$ is watermark-on minus watermark-off accuracy in percentage points. Confidence intervals are paired 90\% intervals. Both TOST tests pass at their stated margins.}
\label{tab:utility}
\begin{tabular}{lrrrr}
\toprule
Dataset & $\Delta$ & 90\% CI & Margin & TOST $p$\\
\midrule
GNSS & $-0.04$ & $[-0.585,0.505]$ & 1.0 & 0.0052\\
CIFAR-10 & $+0.24$ & $[-0.018,0.498]$ & 0.5 & 0.0489\\
\bottomrule
\end{tabular}
\end{table}

\textbf{Matched global-model comparison.} We pair the clean rows of the ten per-seed attribution-attack tables with the corresponding watermark-off controls, using the same dataset and seed and differing only in $\lambda_{\mathrm{wm}}$. The test uses the sample standard deviation of paired differences and Student's $t$ distribution with nine degrees of freedom. Table~\ref{tab:utility} reports the mean, 90\% interval, and the larger one-sided TOST $p$-value. Both datasets establish equivalence at the specified margins. This result concerns episodic accuracy at $N_c=10$ and does not establish equality of feature geometry or dispatched-copy utility. The reproduction artifact records every input path, file hash, and paired difference.

\subsection{Identity-Survival Diagnostics}
\label{sub:tightness}

The quadrature approximation uses the measured pretrained scale norm $G$ ($7.32$ on CIFAR-10, $4.73$ on GNSS), with no fitted coefficient. Its per-bit prediction and the observed self-agreement both decrease with identity load in Table~\ref{tab:survival}. The observations exceed the approximation in every cell. Repeated embedding may contribute to that gap, but this comparison does not establish the geometry or conditional independence needed by Theorem~\ref{thm:perbit}.

At $N_c=40$, GNSS and CIFAR-10 have similar self-agreement ($0.77$ and $0.76$), but attribution rates of $0.87$ and $1.00$. Mean agreement therefore does not determine attribution: the decoding radius and minimum codeword separation in Theorem~\ref{thm:attr} control the competing decisions. The smaller GNSS scale norm improves the quadrature prediction and cannot explain its earlier attribution loss. Above $\rho=1$, the projection system is overcomplete, but rank alone does not force decoding failure. These data establish the dataset difference without identifying its cause.

\subsection{Adversarial Attack Protocol}
\label{sub:attack-protocol}

\textbf{Model-modification attacks.} Attack~(1) prunes 5--90\% of the smallest-magnitude BN~$\boldsymbol{\gamma}$ values. Its structured variant removes entire channels, including their $\gamma$, $\beta$, and running statistics. Attack~(2) adds Gaussian noise $\mathcal{N}(0,\sigma\,\mathrm{std}(\boldsymbol{\gamma}))$ to all BN~$\boldsymbol{\gamma}$, with $\sigma\in\{0.01,0.05,0.1,0.5,1.0,2.0\}$. Attack~(3) uniformly quantizes BN~$\boldsymbol{\gamma}$ to $b$ bits, where $b\in\{16,8,6,4,3,2\}$. Attack~(4) combines pruning and quantization in six configurations.

\textbf{Targeted attacks.} Attack~(5) uses norm-constrained PGD to maximize a bit-flip objective. Each step is projected into an $\epsilon$-ball around the watermarked $\boldsymbol{\gamma}$, using an $\ell_\infty$ clamp or $\ell_2$ rescaling with $\epsilon\in\{0.05,0.1,0.2\}$. Attack~(6) resets the BN~$\boldsymbol{\gamma}$ values in selected ResNet-18 stages to their default $1.0$. We test resets from a single stage through all BN layers.

\textbf{Training-based attacks.} Attack~(7) fine-tunes the model for $10$--$500$ episodic ProtoNet episodes without the watermark term. Attack~(8) distills a freshly initialized ResNet-18 student for $5$--$100$ epochs using the same data and optimizer. The feature-matching objective combines mean-squared embedding error with a Kullback--Leibler term on temperature-softened embedding coordinates ($T=4$), following the hint-matching approach of~\cite{romero2015fitnets}. Both terms align feature coordinates without matching task logits.

The distillation-class sweep distinguishes this objective from soft-label distillation~\cite{hinton_vinyals_dean}. It tests feature matching, cross-architecture transfer to ResNet-34, logit-only distillation, and feature isometry. The isometry composes a frozen backbone with a random orthogonal map, preserving the output function while changing feature coordinates (Sec.~\ref{sec:eval_robustness}).

\textbf{Insider and deployment stressor budgets.} Own-row erasure targets the recipient's tracing row on each carrier. The negation strategy writes the opposite bits, while fresh erasure writes an independent random row. The hybrid combines weight-overlay negation with fresh feature-carrier erasure. We test $300$, $600$, and $1200$ attack episodes, corresponding to $1$--$4{\times}$ the matched split budget. For GNSS receiver covariate shift, we perturb the GNSS front-end with a $+3$\,dB gain offset and sweep per-station SNR over $5$--$20$\,dB.